\documentclass[lettersize,journal]{IEEEtran}
\usepackage{amsmath,amsfonts}
\usepackage{algorithmic}
\usepackage{algorithm}
\usepackage{array}
\usepackage[caption=false,font=normalsize,labelfont=sf,textfont=sf]{subfig}
\usepackage{textcomp}
\usepackage{stfloats}
\usepackage{url}
\usepackage{verbatim}
\usepackage{graphicx}
\def\BibTeX{{\rm B\kern-.05em{\sc i\kern-.025em b}\kern-.08em
    T\kern-.1667em\lower.7ex\hbox{E}\kern-.125emX}}
\usepackage{balance}
\usepackage{amssymb,amsthm}
\usepackage{enumitem}
\usepackage[normalem]{ulem}
\usepackage{mathrsfs}

\newtheorem{assumption}{Assumption}
\newtheorem{theorem}{Theorem}

\newtheorem{remark}{Remark}
\newtheorem{lemma}{Lemma}
\newtheorem{corollary}{Corollary}

\newcommand{\aresref}[1]{ARES}

\usepackage{color}  
\usepackage{hyperref}
\hypersetup{
    colorlinks=true, 
    linktoc=all,     
    linkcolor=blue,  
}

\usepackage{titlesec}
\usepackage{chngcntr}

\newcommand{\setupautartappendix}{%
  \counterwithin{equation}{section}
  \renewcommand{\theequation}{\Alph{section}.\arabic{equation}}
  \renewcommand{\thesubsection}{\Alph{section}.\arabic{subsection}}

  \titleformat{\section}[block]
    {\normalfont\large\bfseries}
    {\Alph{section}}{1em}{}
  \titleformat{\subsection}[block]
    {\normalfont\normalsize\itshape}
    {\Alph{section}.\arabic{subsection}.}{0.6em}{}
}

\begin{document}
\title{An Adaptive Method for Contextual Stochastic Multi-armed Bandits with Rewards Generated by a Linear Dynamical System}
\author{Jonathan Gornet, \IEEEmembership{Student Member, IEEE}, Mehdi Hosseinzadeh, \IEEEmembership{Member, IEEE},\\ and Bruno Sinopoli, \IEEEmembership{Fellow, IEEE}%
\thanks{This research has been supported by Army Research Office under award number W911NF-20-S-0009.}
\thanks{J. Gornet and B. Sinopoli are with the Department of Electrical and Systems Engineering, Washington University in St. Louis, St. Louis, MO 63130, USA (email: jonathan.gornet@wustl.edu; bsinopoli@wustl.edu).}
\thanks{M. Hosseinzadeh is with the School of Mechanical and Materials Engineering, Washington State University, Pullman, WA 99164, USA (email: mehdi.hosseinzadeh@wsu.edu).}
}

\maketitle

\begin{abstract}
Online decision-making can be formulated as the popular stochastic multi-armed bandit problem where a learner makes decisions (or takes actions) to maximize cumulative rewards collected from an unknown environment. This paper proposes to model a stochastic multi-armed bandit as an unknown linear Gaussian dynamical system, as many applications, such as bandits for dynamic pricing problems or hyperparameter selection for machine learning models, can benefit from this perspective. Following this approach, we can build a matrix representation of the system's steady-state Kalman filter that takes a set of previously collected observations from a time interval of length $s$ to predict the next reward that will be returned for each action. This paper proposes a solution in which the parameter $s$ is determined via an adaptive algorithm by analyzing the model uncertainty of the matrix representation. This algorithm helps the learner adaptively adjust its model size and its length of exploration based on the uncertainty of its environmental model. The effectiveness of the proposed scheme is demonstrated through extensive numerical studies, revealing that the proposed scheme is capable of increasing the rate of collected cumulative rewards.
\end{abstract}

% \begin{IEEEkeywords}
% Hyperparameter Optimization, Bayesian Optimization, Stochastic Multi-armed Bandits, Kalman filters, Stochastic Dynamical Systems
% \end{IEEEkeywords}

\section{Introduction}
The stochastic multi-armed bandit problem focuses on modeling online learning processes to understand decision-making under uncertainty with the aim of providing a simplified framework to situations encountered frequently. The problem considers a learner interacting with an environment, where the learner chooses an action for which the environment in return gives a reward that is sampled from a probability distribution. Performance is usually measured in terms of cumulative regret, defined as the cumulative difference between the highest reward that could be given for a round and the reward given for the chosen action. Since the learner typically lacks prior knowledge about the environment \cite{lattimore2020bandit}, the learner needs to balance between exploring each action to extract information and committing to an action to maximize cumulative reward. A popular adaptive exploration method that has been introduced where the distributions are assumed to be stationary is the Upper Confidence Bound (UCB) method \cite{agrawal1995sample}.

In the case of non-stationary stochastic multi-armed bandits, the distributions of the rewards can change over time \cite{lattimore2020bandit}. Therefore, learners need to understand how its previous information of the environment may misrepresent the current reward distributions. Most of the existing methods have focused on \textit{piece-wise stationary} cases, where the distribution shifts abruptly \cite{garivier2008upper}. In this context, presented algorithms are commonly approached as a forgetting versus memorizing trade-off, where either the estimates of the reward distributions are forgotten when a distributional shift is detected \cite{hartland2007change,liu2018change,cao2019nearly,mellor2013thompson} or the estimates of the reward distributions use measurements that are within a sliding window \cite{garivier2008upper}. Some non-stationary methods focused on the impact of distributional change in the form of \textit{variational budget}, which is a known constant that bounds the cumulative changes in the distributions \cite{wei2021nonstationary,besbes2014stochastic}. In contrast to piece-wise stationary cases, the distribution may change slowly, which are referred as slowly-varying multi-armed bandits \cite{wei2018abruptly}. For the slowly-varying multi-armed bandits, \cite{slivkins2008adapting} modeled rewards of each arm as Brownian motion. This approach has been extended to the scenarios where the rewards are generated by a partially observable stochastic linear dynamical system \cite{9993316}.

Another variation of the stochastic multi-armed bandit is the restless bandit, introduced in 1988 \cite{whittle1988restless}. In the restless bandit, each action has its own Markov chain that usually consists of discrete states. For every action the learner chooses, the learner observes the Markov chain's state variable and the reward, which is a function of the state variable. Previous results in the discrete state-space Markov chain that use an approach similar to UCB are \cite{liu2012learning,tekin2012online,ortner2014regret,dai2011non,wang2020restless}. Authors in \cite{jung2019regret} use Thompson sampling, i.e., sampling parameters based on \textit{a priori} distribution of Markov chain, for action selection. To the best of our knowledge, \cite{wang2020restless} is the most recent result. In the work, the algorithm Restless-UCB first learns the Markov chain for each action for a set interval. After the learning interval phase, Restless-UCB then searches for the optimal action sequence using the learned Markov chain perturbed by the magnitude of its model error. 

The last variation of the stochastic multi-armed bandit problem relevant to this paper is the stochastic contextual bandit. In the contextual bandit, the learner observes a context in addition to the reward, where it is assumed that the reward is a function of the action and of the observed context \cite{lattimore2020bandit}. A specific version of the stochastic contextual bandit is the stochastic linear bandit where the reward is a linear combination of an unknown linear parameter and a known action vector \cite{abe1999associative}. In this context, \cite{NIPS2011_e1d5be1c} uses an optimism-based exploration strategy for identifying the optimal action. The main focus in \cite{bastani2021mostly} is an analysis of exploration-free greedy methods, i.e. directly sampling actions based on what the learner expects to return the highest reward, in comparison to optimism-based exploration. An algorithm developed in \cite{kuroki2024best} addresses environments where the rewards and contexts can either be stochastic or adversarial (adversarial is when the environment chooses the sequence of rewards and contexts prior to its interaction with the learner). In some cases, the context may be partially observed. In this setting,  \cite{parkbalancing} uses a Thompson sampling method by sampling the unknown linear parameter from a prior probability distribution. 

For this work, we consider the rewards generated by a Linear Gaussian Dynamical System (LGDS), and we envision it can be utilized to address challenges in many applications, including dynamic pricing problem \cite{mueller2019low,agrawal2021dynamic}, controller selection for time-varying stochastic linear dynamical system  \cite{pmlr-v211-gradu23a,NEURIPS2021_856b503e}, opportunistic spectrum access in wireless networks \cite{tekin2011online,tekin2012online}, portfolio optimization \cite{shen2015portfolio,huo2017risk,gornet2022stochastic}, and hyperparameter optimization \cite{li2017hyperband,gornet2025hypercontrollerhyperparametercontrollerfast} or algorithm selection \cite{gagliolo2010algorithm} for machine learning models. Each of the applications above can be modeled as a dynamical system as previous efforts have made a similar assumption such as \cite{agrawal2021dynamic} for dynamic pricing, \cite{parker2020provably,gornet2025hypercontrollerhyperparametercontrollerfast} for hyperparameter optimization, and \cite{pmlr-v211-gradu23a} for controller selection. Finally, many of the examples above have a context available, i.e. the computed loss from a cost function in a machine learning model can be viewed as a context.  

Building upon the authors' previous work \cite{9993316}, this paper develops an adaptive algorithm that maximizes cumulative rewards sampled from an unknown stochastic linear dynamical system. The key difference in our proposed environment in comparison to the restless bandit and stochastic linear bandit are as follows. First, the LGDS, which is a Markovian system, has a continuous state-space in comparison to the discrete state-space in restless bandits. Second, although we assume that the context is a partial observation of the LGDS state which shares similarities with previous literature in stochastic linear bandits, the difference is that the contexts are also time correlated.

In this paper, we develop a method that adaptively chooses the number of previous observations to be taken into account, and the length of exploration based on the uncertainty of the model that the learner has learned for the environment. This method allows the learner explore the environment adaptively. This addresses the weakness of Phased Initial Exploration of the System (PIES) \cite{9993316}, where the window size and exploration length of the environment are fixed, which can lead to a situation where the exploration length is longer than the horizon length. Therefore, this paper provides a perspective on how learners should integrate model selection with its decisions.

This paper is organized as follows. Section \ref{sec:Problem_Framework} states the problem. In Section \ref{sec:Modeling_the_System}, a methodology to model and predict rewards is presented. Section \ref{sec:controlling_model_error} discusses how model size and its errors can impact performance. Section \ref{sec:Bandit_Strategy} uses the learned model to develop a strategy to maximize cumulative reward over a horizon. This section also analyzes the regret of the developed strategy. Section \ref{sec:Numerical_Simulation} has numerical results. Finally, Section \ref{sec:Conclusion} concludes the paper and provides future directions.

\smallskip
\noindent\textbf{Notation:} For any $x\in\mathbb{R}^n$ and $y\in\mathbb{R}^n$, we have the inner product $\left\langle x, y\right\rangle = x^\top y \in \mathbb{R}$. For the natural number $i \in \mathbb{N}$, we have $[i]=\{1,2,\dots,i\}$. The distribution $\mathscr{N}\left(\mu,\Sigma\right)$ is a normal distribution with a mean $\mu \in \mathbb{R}^d$ and a covariance of $\Sigma \succeq \mathbf{0} \in \mathbb{R}^{d \times d}$. The norm $\left\Vert x \right\Vert_{Q} \triangleq \sqrt{x^\top Q x}$ is the weighted $\ell$-2 norm where $Q \succ \mathbf{0}$.  

% {\color{red}Comment: There  are other notations that could be included here. For instance, the normal distribution $\mathcal{N}$, weighted norm $\left\Vert x\right\Vert_Q^2$, etc. }

\section{Problem Statement}\label{sec:Problem_Framework}
Suppose that for $k$ given actions $c_a \in\mathscr{A} \triangleq \{c_a \in \mathbb{R}^d \mid a = 1,2,\dots,k\}$, the reward $X_t$ is sampled from the following linear Gaussian dynamical system
\begin{equation}\label{eq:Linear_System}
\begin{cases}
z_{t+1} & = \Gamma z_t + \xi_t,~z_0 \sim \mathscr{N}(\mathbf{0},\Sigma_0) \\
\theta_t & = C_\theta z_t + \phi_t \\
X_t & = \langle c_a ,z_t \rangle + \eta_t
\end{cases}, 
\end{equation}
where $z_t \in \mathbb{R}^d$ is the state of the system and $\Gamma \in \mathbb{R}^{d \times d}$ is the system state matrix. For each round $t=1,2,\dots,n$ and $n>d$, the learner observes the reward $X_t\in\mathbb{R}$ based on the chosen action and the context $\theta_t\in\mathbb{R}^m$. The context $\theta_t$ is a dynamically changing value with respect to the LGDS state variable $z_t$ that the learner always observes and its observation matrix $C_{\theta} \in \mathbb{R}^{m \times d}$ is constant. The noises $\xi_t \in \mathbb{R}^d$, $\phi_t \in \mathbb{R}^m$, and $\eta_t \in \mathbb{R}$ are i.i.d. normally distributed, i.e., $\xi_t \sim \mathscr{N}(\mu ,Q)$ with $Q\succeq\mathbf{0}$, $\phi_t \sim \mathscr{N}(\mathbf{0},R)$ with $R \succ \mathbf{0}$, and  $\eta_t \sim \mathscr{N}(0,\sigma_\eta^2)$ with $\sigma_\eta > 0$. For the noise term $\eta_t \in \mathbb{R}$, we will assume that each action $c_a \in \mathcal{A}$ has its own independently generated noise $\eta_t \in \mathbb{R}$.

\begin{assumption}
The matrices $\Gamma$, $C_\theta$, $Q$, $R$, $\Sigma_0$, vector $\mu$, and scalar $\sigma_\eta$ are unknown. The dimension $d$ is unknown, but dimension $m$ is known as it is the dimension of the observed context. Also, the vector $c_a$ is unknown. Note that for notation, given that there are $k$ vectors $c_a \in \mathscr{A}$, we denote $a \in [k]$ to be which vector $c_a$ is chosen.
\end{assumption}

\begin{assumption}\label{assum:observable}
The matrix $\Gamma$ is marginally stable for all $t \geq 0$; that is, its eigenvalues are within and on the unit circle. Also, the pair $(\Gamma,C_\theta)$ is observable and the pair $(\Gamma,Q^{1/2})$ is controllable. 
\end{assumption}

The goal of the learner is to maximize cumulative reward over a finite time horizon $n$. To do so, \textit{pseudo-regret} analysis is used. The \textit{pseudo-regret} is defined as the cumulative, over all rounds, difference between the highest reward (denoted as $X_t^\ast$) and the reward for the chosen action at time $t$, i.e.,

\vspace{-1em}
\begin{equation}\label{eq:pseudo-regret}
R_n = \sum_{t = 1}^n X_t^\ast-X_t. 
\end{equation}
\vspace{-1em}

For this paper, we will use \textit{pseudo-regret} \eqref{eq:pseudo-regret} instead regret, where regret is defined as the cumulative, over all rounds, \textit{expected} difference between the highest reward and the reward for the chosen action at time $t$, i.e.,

\vspace{-1em}
\begin{equation}\label{eq:regret}
R_n = \sum_{t = 1}^n \mathbb{E}\left[X_t^\ast-X_t\right]. 
\end{equation}
\vspace{-1em}

This allows us to consider variance in regret when designing the strategy and has been proposed as a better metric in comparison to regret \cite{lattimore2020bandit}. For the duration of the text, we will refer to \textit{pseudo-regret} as regret.  

\begin{remark}
    We will define the optimal action $a^*$ to be the action $c_a \in\mathcal{A}$ that aligns most closely with the LGDS state variable $z_t \in \mathbb{R}^d$, i.e. $a^* \triangleq \arg\max_{a \in [k]}\left\langle c_a,z_t\right\rangle$. 
\end{remark}

\section{Matrix Representation of the System}\label{sec:Modeling_the_System}
If the learner knew parameters of system  \eqref{eq:Linear_System}, then they could use the Kalman filter (which is the optimal prediction in the mean-squared sense) expressed in \eqref{eq:Kalman_Filter} to predict the state $z_t$, and consequently the reward $X_t$, for each action $a \in [k]$:
\begin{equation}\label{eq:Kalman_Filter}
\left\{
\begin{array}{ll}
\hat{z}_{t+1|t} & = \Gamma \hat{z}_{t|t} + \mu \\ P_{t+1|t} & = \Gamma P_{t|t} \Gamma^\top + Q \\
K_t & = P_{t|t-1} C_{\theta}^\top (C_{\theta} P_{t|t-1}C_{\theta}^\top + R)^{-1} \\
\hat{z}_{t|t} & = \hat{z}_{t|t-1} + K_t (\theta_t - C_{\theta} \hat{z}_{t|t-1}) \\
P_{t|t} & = P_{t|t-1} - K_t C_{\theta} P_{t|t-1} \\
\hat{X}_{t|t-1} & = \langle c_a,\hat{z}_{t|t-1}\rangle 
\end{array} 
\right.,
\end{equation}
where 
\begin{align}
    \hat{z}_{t|t} & \triangleq \mathbb{E}\left[z_t|\mathscr{F}_t\right] \nonumber \\
    P_{t|t} & \triangleq \mathbb{E}\left[(z_t-\hat{z}_{t|t})(z_t-\hat{z}_{t|t})^\top|\mathscr{F}_t\right] \nonumber, 
\end{align}
and $\mathscr{F}_t$ is the sigma algebra generated by previous contexts $\theta_0,\dots,\theta_t$. Since the pair $(\Gamma, C_\theta)$ is observable (see Assumption \ref{assum:observable}), the Kalman gain matrix $K_t$ converges in few steps \cite{tsiamis2019finite}. Therefore, the following steady-state Kalman filter can be used:
\begin{equation}\label{eq:ss_Kalman_Filter}
\left\{\begin{array}{ll}
\hat{z}_{t+1} & = \Gamma \hat{z}_{t} + \mu + \Gamma K (\theta_t - C_\theta \hat{z}_{t})\\
\hat{X}_t & = \langle c_a,\hat{z}_{t}\rangle 
\end{array}\right.,
\end{equation}
where
\begin{align}
\left\{
 \begin{array}{ll}
K & \triangleq PC_\theta^\top (C_\theta PC_\theta^\top + R)^{-1}\\
\hat{z}_t & \triangleq \hat{z}_{t|t-1}\\
\hat{X}_t & \triangleq \hat{X}_{t|t-1} = \mathbb{E}[X_t \mid \mathscr{F}_{t-1}]
\end{array} \right., \label{eq:reward_filtration}
\end{align}
with $P$ being the solution to the Riccati equation $P = \Gamma  P\Gamma ^\top + Q - \Gamma PC_\theta^\top (C_\theta PC_\theta^\top + R)^{-1} C_\theta P\Gamma^\top$.

Since system parameters are unknown, taking the inspiration from \cite{tsiamis2019finite}, we develop a method to learn the matrix representation of the Kalman filter. This matrix will be a function of the matrices and vectors defined in the steady-state Kalman filter given in \eqref{eq:ss_Kalman_Filter} and \eqref{eq:reward_filtration}. Let $s>0$ be the horizon length of how far we look into the past. From \eqref{eq:ss_Kalman_Filter}, we can express the reward prediction $\hat{X}_t$ using the observations from $t-s$ to $t-1$, as follows:
\begin{multline}\label{eq:backstep_s_steps}
    \hat{X}_t = \left\langle c_a, \left(\Gamma-\Gamma K C_{\theta}\right)^{s-1} \Gamma K \theta_{t-s} \right\rangle + \cdots \\ + \left\langle c_a, \Gamma K\theta_{t-1} \right\rangle + \sum_{\tau=1}^s \left\langle c_a, \left(\Gamma - \Gamma K C_\theta\right)^\tau \mu \right\rangle \\ + \left\langle c_a, \left(\Gamma - \Gamma K C_\theta\right)^s \hat{z}_{t-s} \right\rangle. 
\end{multline}

Since $X_t = \hat{X}_t + \varepsilon_a^t$, where $\varepsilon_a^t \in \mathbb{R}$ is defined to be 
\begin{equation}
    \varepsilon_a^t \triangleq \langle c_a, z_t - \hat{z}_t \rangle + \eta_t \sim \mathscr{N}(0,c_a^\top P c_a + \sigma_\eta^2), \label{eq:residual_error}
\end{equation}
we can use the expression of $\hat{X}_t$ given in \eqref{eq:backstep_s_steps} to express the reward $X_t$ as follows:
\begin{multline}
    X_t = \left\langle c_a, \left(\Gamma-\Gamma K C_{\theta}\right)^{s-1} \Gamma K \theta_{t-s} \right\rangle + \cdots \\ + \left\langle c_a, \Gamma K\theta_{t-1} \right\rangle + \sum_{\tau=1}^s \left\langle c_a, \left(\Gamma - \Gamma K C_\theta\right)^\tau \mu \right\rangle \\ + \left\langle c_a, \left(\Gamma - \Gamma K C_\theta\right)^s \hat{z}_{t-s} \right\rangle +  \varepsilon_a^t \nonumber. 
\end{multline}

We define vectors $G_{a}\left(s\right)$ for each $a \in [k]$ and $\Theta_t$ as follows:
\begin{align}\label{eq:G_a}
G_{a}\left(s\right) \triangleq\Bigg[&
c_a^\top\left(\Gamma-\Gamma K C_{\theta}\right)^{s-1} \Gamma K~~\cdots ~~ c_a^\top\Gamma K \nonumber \\
& \sum_{\tau = 1}^s \left\langle c_a, \left(\Gamma-\Gamma K C_{\theta}\right)^{\tau} \mu \right\rangle \Bigg]^\top \in \mathbb{R}^{(ms+1) \times 1}, \\
\Theta_{t}\triangleq \Big[&
\theta_{t-s}^\top~~\cdots~~\theta_{t-1}^\top~~1\Big]^\top \in \mathbb{R}^{(ms+1)\times 1}.
\end{align}

Using $G_{a}\left(s\right)$ and $\Theta_{t}$ defined above, the reward $X_t$ can be expressed as
\begin{align}
X_{t} & = G_{a}\left(s\right)^\top \Theta_{t} +\beta_a^t +  \varepsilon_a^t, \label{eq:Matrix_Form_2}
\end{align}
where 
\begin{equation}
\beta_a^t \triangleq \left\langle c_a, \left( \Gamma-\Gamma K C_{\theta}\right)^s \hat{z}_{t-s} \right\rangle. \label{eq:residual_error_beta}
\end{equation}

\begin{remark}
According to \eqref{eq:Matrix_Form_2}, at round $t$, given the set of contexts from round $t-s$ up to round $t-1$ (i.e., $\{\theta_{t-s},\dots,\theta_{t-1}\}$), one can predict the reward $X_t$, where the prediction accuracy is impacted by the terms $\beta_a^t$ and $\varepsilon_a^t$ given in \eqref{eq:residual_error_beta} and \eqref{eq:residual_error}, respectively. 
\end{remark}

To continue the discussion, we impose the the following assumptions.

\begin{assumption}\label{assum:G_a_Assumption}
    According to discussions in \cite{NIPS2011_e1d5be1c}, it is reasonable to assume that there exists a known constant $B_G$ such that $\Vert G_{a}\left(s\right) \Vert_2 \leq B_G$ for all $a \in [k]$, $s \in [s_N]$.
\end{assumption}

\begin{assumption}\label{assum:action_bound}
    There exists a known constant $B_c > 0$ such that $\left\Vert c_a \right\Vert_2 \leq B_c,~a\in[k]$. 

\end{assumption}

\begin{assumption}\label{assum:noise_bound}
We assume that there exists a known constant $B_R > 0$ such that for $\varepsilon_a^t$ given in \eqref{eq:residual_error}, the variance is bounded such that for any $a \in [k]$, i.e., $c_a^\top P c_a + \sigma_\eta^2 \leq B_R^2$. 
\end{assumption}

\begin{remark}\label{remark:MagnitudeofBeta}
From Assumption \ref{assum:action_bound}, we have $\left\vert\beta_a^t\right\vert\leq B_c\left\Vert\left(\Gamma - \Gamma K C_\theta\right)^s\right\Vert_2\left\Vert\hat{z}_{t-s}\right\Vert_2$. Thus, since $\Gamma - \Gamma K C_{\theta}$ is Schur by construction, $\left\Vert\left(\Gamma - \Gamma K C_\theta\right)^s\right\Vert_2$ decreases as $s$ increases which implies that $\left\vert\beta_a^t\right\vert$ decreases as $s$ increases. %{\color{red}I am not sure if this is correct, as we have $\left\Vert\hat{z}_{t-s}\right\Vert_2$.}
\end{remark}

The expression of the reward $X_t$ given in \eqref{eq:Matrix_Form_2} consists of unknown terms $G_a\left(s\right)$, $\beta_a^t$, and $\varepsilon_a^t$ and a known term $\Theta_t$. Since $\varepsilon_a^t$ is a zero-mean normally distributed random variable, then \eqref{eq:Matrix_Form_2} is a linear model with an unknown bias term $\beta_a^t$. Note that we will analyze the impact of the bias term later in this paper. Therefore, linear least squares can be used for identifying $G_a\left(s\right)$. Let $\mathscr{T}_{a} = \{t_1,\dots,t_{N_a}\}$ be the set of rounds that action $a \in [k]$ is chosen. Then, according to \eqref{eq:Matrix_Form_2}, the reward $X_t$ can be expressed as:
\begin{equation}\label{eq:linear_model_time_varying}
    \mathbf{X}_{\mathscr{T}_a} = G_{a}\left(s\right)^\top \mathbf{O}_{\mathscr{T}_a} + \mathbf{B}_{\mathscr{T}_a} + \mathbf{E}_{\mathscr{T}_a},
\end{equation}
where
\begin{align}
\left\{
\begin{array}{l}
\mathbf{X}_{\mathscr{T}_a} \triangleq \begin{bmatrix}
    X_{t_1} & \dots & X_{t_{N_a}}
    \end{bmatrix} \in \mathbb{R}^{1 \times N_a}\\
\mathbf{O}_{\mathscr{T}_a} \triangleq \begin{bmatrix}
    \Theta_{t_1} & \dots & \Theta_{t_{N_a}}
    \end{bmatrix} \in \mathbb{R}^{(ms+1) \times N_a}\\
\mathbf{B}_{\mathscr{T}_a} \triangleq \begin{bmatrix}
    \beta_{a}^{t_1} & \dots & \beta_{a}^{t_{N_a}}
    \end{bmatrix} \in \mathbb{R}^{1 \times N_a}\\
\mathbf{E}_{\mathscr{T}_a} \triangleq \begin{bmatrix}
    \varepsilon_{a}^{t_1} & \dots & \varepsilon_{a}^{t_{N_a}}
    \end{bmatrix} \in \mathbb{R}^{1 \times N_a}
    \end{array} 
    \right.. \label{eq:matrix_definitions}
 \end{align}

Let $V_a^t\left(s\right) =\lambda I + \mathbf{O}_{\mathscr{T}_a} \mathbf{O}_{\mathscr{T}_a}^\top$, where $\lambda > 0$, and as a result, $V_a^t\left(s\right)$ is invertible. Then, a least squares approximation of the matrix $G_{a}\left(s\right)$ can be computed as:
\begin{align}
\hat{G}_a^t\left( s\right)^\top =\mathbf{X}_{\mathscr{T}_a} \mathbf{O}_{\mathscr{T}_a}^\top V_a^t\left( s\right)^{-1}\label{eq:identify_2}, 
\end{align}
where $t$ indicates that \eqref{eq:identify_2} is the identified $G_a\left(s\right)$ at round $t$.

In comparison with the Kalman filter approach, only a finite set of size $s$ of previous contexts are required to predict the reward $X_t$. However, the predictions will be biased by the term $\beta_a^t$  whose magnitude depends on the value $s$ and the matrix $\Gamma - \Gamma K C_\theta$.

\begin{remark}
    The sigma algebra $\mathscr{F}_t$ ignores the sequence of rewards $X_0,\dots,X_t$ on the state prediction $\hat{z}_t$. However, if we consider adding the sequence of the rewards $X_0,\dots,X_t$ and the sequence of contexts $\theta_0,\dots,\theta_t$ to the sigma algebra $\mathscr{F}_t$, then the steady-state Kalman filter does not exist, as the error covariance matrix does not converge since the action $c_a \in \mathcal{A}$ changes each round $t$. The complexity of this problem is reviewed in \cite{gornet2024restless}. 
\end{remark}

\section{Controlling Model Error}\label{sec:controlling_model_error}

Assuming that the reward $X_t$ is sampled from system \eqref{eq:linear_model_time_varying}, the learner will use \eqref{eq:identify_2} to predict which action $a \in [k]$ will return the highest reward $X_t$. The learner needs to consider the accuracy of the predictions, as the bias term $\beta_a^t$ may have a significant impact on the predictions obtained via \eqref{eq:identify_2}.

The accuracy of the reward prediction is impacted by the number of times an action $a \in [k]$ has been chosen denoted as $N_a$ (as discussed in \cite{tsiamis2019finite}, the larger the sample set, the lower the model error) and the window length $s>0$ (as discussed in Remark \ref{remark:MagnitudeofBeta}); see Fig. \ref{figure:parameter_values} for an illustration of $N_a$ and $s$. Therefore, to optimize the accuracy of the reward prediction, we propose to control the values $N_a$ and $s$, which will be discussed next.

\begin{figure}[ht]
\centering
\includegraphics[width=\linewidth]{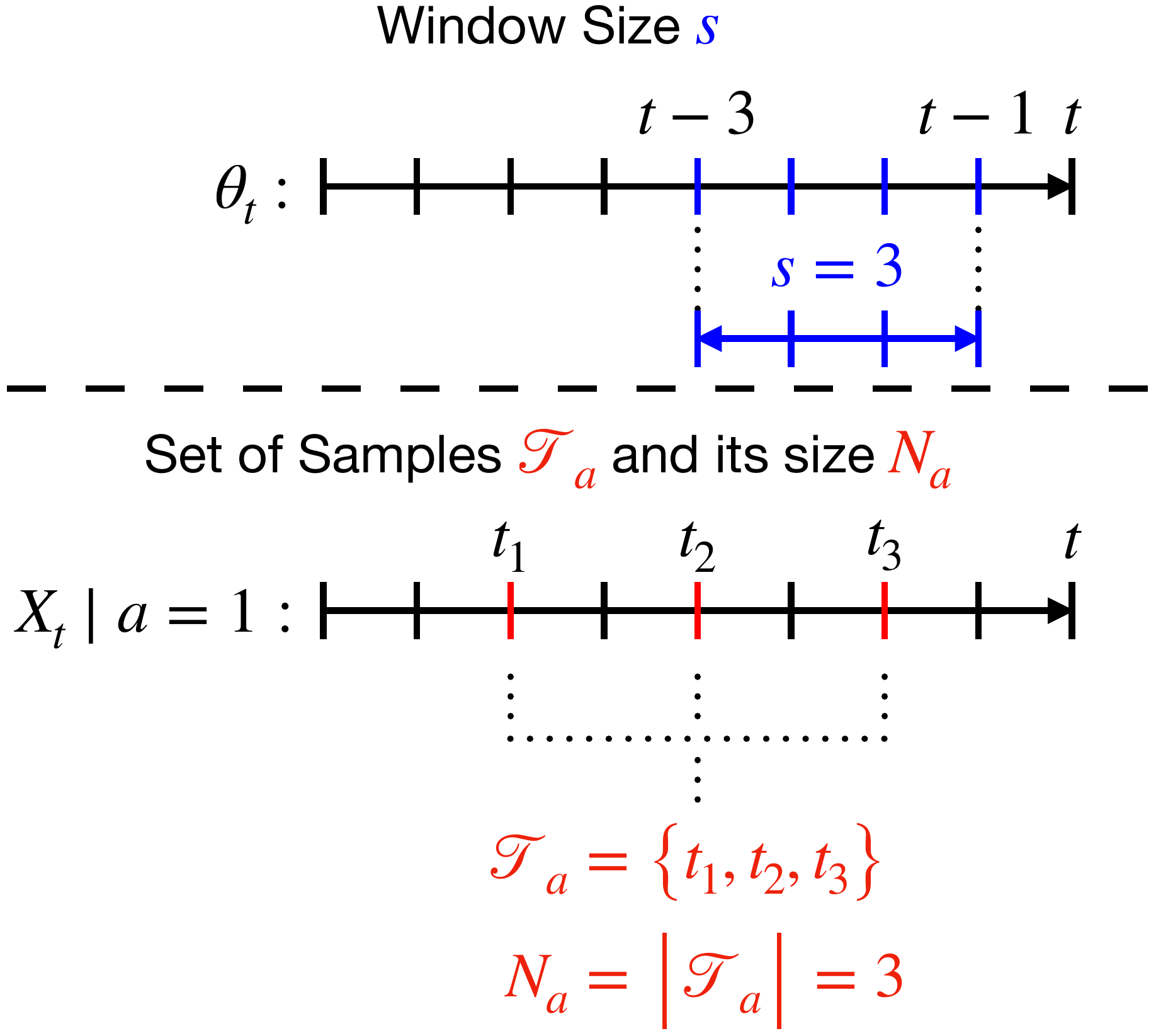}
\caption{The parameter $s$ is the sliding window size of contexts $\theta_{t-s}$ to $\theta_{t-1}$ that are used for predicting the reward $X_t$. For example, if $s = 3$, then at round $t$ the learner uses contexts $\theta_{t-3}$ to $\theta_{t-1}$ for predicting the reward $X_t$. The set $\mathscr{T}_a$ has all the rounds $t_i$ where action $a \in [k]$ is chosen and $N_a$ is the number of times action $a \in [k]$ that has been chosen by the learner. For example, if the learner choose action $a = 1$ at rounds $t_1$, $t_2$ and $t_3$, then set $\mathscr{T}_1= \{t_1,t_2,t_3\}$ and $N_1 = 3$.}
\label{figure:parameter_values}
\end{figure}

\begin{remark}
We will control the parameters $N_a$ and $s$ separately in a hierarchical manner. We leave the study of their possible impacts on each other to future work. 
\end{remark}

\subsection{Perturbation Method}\label{sec:UCB}
Taking inspiration from \cite{agrawal1995sample}, 
this subsection provides a way of controlling $N_a$ by using a perturbation signal $u_a$. 

The learner chooses actions $a \in [k]$ with the highest reward prediction. Due to model error, the learner may constantly choose a sub-optimal action $a \in [k]$ without exploring other actions. To prevent this situation, a perturbation input signal $u_a \in \mathbb{R}$ is added such that the action the learner chooses is based on the following optimization problem
\begin{equation}\label{eq:perturbation_method}
    \max_{a \in [k]} \,\hat{G}_a^t\left( s\right)^\top \Theta_t + u_a.
\end{equation}

We will design $u_a$ such that the regret is optimized. To understand how the signal $u_a$ impacts performance, we will consider the probability of choosing the optimal action $a_t^*$ using $u_a$; the following steps will be used for understanding the impact: \\
\noindent\textbf{$\bullet$~STEP 1:} Providing a bound on the model error
\begin{equation*}
    \left\Vert \hat{G}_a^t\left(s\right) - G_{a}\left(s\right) \right\Vert_{V_a^t\left(s\right)}.
\end{equation*}
This step provides Theorem \ref{theorem:model_error} which is written below. All proofs can be found in the appendix of our Arxiv version \cite{arXiv_proofs}. 
\begin{theorem}\label{theorem:model_error}
Let $\hat{G}_a^t\left(s\right)$ be identified as in \eqref{eq:identify_2} and the reward $\mathbf{X}_{\mathscr{T}_a}$ is as in \eqref{eq:linear_model_time_varying}. Then, there exists $\delta\in\left(0,1/2\right)$ such that the following inequality is satisfied with a probability of at least $1-2\delta$:
\begin{equation}
\left\Vert \hat{G}_a^t \left( s\right) - G_{a}\left(s\right) \right\Vert_{V_a^t\left(s\right)} \leq b_a^t \left(\delta,s\right), \nonumber
\end{equation}
where 
\begin{multline}\label{eq:probabilistic_bound}
b_a^t \left(\delta,s\right) \triangleq  \sqrt{2B_R^2\log\left(\frac{1}{\delta}\frac{ \det(V_a^t\left(s\right))^{1/2}}{\det(\lambda I)^{1/2}}\right)} \\ 
+ \sqrt{N_a  g_{\Sigma_{z_t}}\left(\delta/N_a\right)} B_{\beta}\left(s\right)\sqrt{\mbox{tr}\left(I - \lambda V_a^t\left(s\right)^{-1}\right)}  \\ 
+ \lambda B_G \sqrt{\mbox{tr}\left(V_a^t\left(s\right)^{-1}\right)} . 
\end{multline}
where $g_{\Sigma_{z_t}}\left(\delta\right)$ is defined to be
\begin{multline}
    g_{\Sigma_{z_t}}\left(\delta\right) \triangleq \mbox{tr}\left(\Sigma_{z_t} \right) + \\ \max \left\{\sqrt{\frac{K^4\left\Vert \Sigma_{z_t} \right\Vert_{HS}^2}{\tilde{c}}\log\left(\frac{2}{\delta}\right)}, 
        \frac{K^2\left\Vert \Sigma_{z_t} \right\Vert_2}{\tilde{c}}\log\left(\frac{2}{\delta}\right)\right\} \nonumber, 
\end{multline}
$\tilde{c} > 0$ is an absolute constant and $K = \sqrt{\left(1/2\right)\log\left(2\right)}$. The term $B_{\beta}\left(s\right)$ is defined to be 
\begin{equation}\label{eq:B_beta_def}
    B_{\beta}\left(s\right)\triangleq B_c \left\Vert \left(\Gamma - \Gamma K C_\theta\right)^s\right\Vert_2. 
\end{equation}
\end{theorem}

The derivations and discussions for this bound are found in Appendix \ref{appendix:step_1}. \\
\noindent\textbf{$\bullet$~STEP 2:} Providing a bound on the error  
\begin{equation*}
    \left\vert\hat{G}_a^t\left(s\right)^\top \Theta_t - G_{a}\left(s\right)^\top \Theta_t\right\vert. 
\end{equation*}

The bound for the prediction error above is Lemma \ref{lemma:optimal_value} which is written below. 

\begin{lemma}\label{lemma:optimal_value}
There exists $\delta\in\left(0,1/2\right)$ such that the following inequality holds with a probability of at least $1-2\delta$:
\begin{multline}\label{eq:InequalityJon}
\left\vert\hat{G}_a^t\left(s\right)^\top \Theta_t - G_{a}\left(s\right)^\top \Theta_t\right\vert\\
\leq b_a^t\left(\delta,s\right)\sqrt{\Theta_{t}^\top V_a^t\left(s\right)^{-1} \Theta_{t}}.
\end{multline} 
\end{lemma}

The derivations and discussions for this bound are found in Appendix \ref{appendix:step_2}. \\
\noindent\textbf{$\bullet$~STEP 3:} Characterizing the impact of the proposed method on the probability of selecting the optimal action $a_t^\ast$.

\subsubsection{STEP 3} \label{subsec:step_3} Since \textbf{STEP 1} and \textbf{STEP 2} have been proven in Appendices \ref{appendix:step_1} and \ref{appendix:step_2}, we can now design the perturbation signal $u_a$. A principle commonly used for bandits is Optimism in the Face of Uncertainty \cite{lattimore2020bandit}. This principle states that the learner should be optimistic of actions it is uncertain of. The amount of optimism a learner has for each action is based on the error of the action's reward prediction. Therefore, we will initially design $u_a$ based on this approach. For analysis, Theorem \ref{theorem:regret} and Corollary \ref{corollary:regret_bound} investigate the impact of the perturbation signal $u_a$ on selecting the optimal action. Proofs for Theorem \ref{theorem:regret} and Corollary \ref{corollary:regret_bound} are found in Appendix \ref{appendix:regret_bounds}. Extended versions of Theorem \ref{theorem:regret} and Corollary \ref{corollary:regret_bound} are provided in Appendix \ref{appendix:regret_bounds}.

\begin{theorem}\label{theorem:regret}
    Assume that $n > 0$ is finite. Let actions $a \in [k]$ be selected based on \eqref{eq:perturbation_method} where 
    \begin{equation}\label{eq:u_a_b}
        u_a^{(\mathbf{b})} = \mathbf{b}_a^t\left(\delta,s\right)\sqrt{\Theta_t^\top V_a^t\left(s\right)^{-1} \Theta_t} ,
    \end{equation}
    where $\mathbf{b}_a^t\left(s\right)$ is defined as
    \begin{align}
        \mathbf{b}_a^t \left(\delta,s\right) \triangleq&  \sqrt{2B_R^2\log\left(\frac{1}{\delta}\frac{ \det(V_a^t\left(s\right))^{1/2}}{\det(\lambda I)^{1/2}}\right)} \nonumber \\ 
        &+ \sqrt{N_a g_{\Sigma_{z_t}}\left(\delta/N_a\right)} \mathbf{B}_{\beta}\sqrt{\mbox{tr}\left(I - \lambda V_a^t\left(s\right)^{-1}\right)}\nonumber\\ 
        &+ \lambda B_G \sqrt{\mbox{tr}\left(V_a^t\left(s\right)^{-1}\right)} \label{eq:b_def}, 
    \end{align}
    and $\mathbf{B}_{\beta} \geq B_{\beta}\left(s\right)$. Regret $\sum_{t=1}^n X_t^* - X_t$ has the following rate that is satisfied with a probability of at least $1 - 9\delta$:
    \begin{multline}\label{eq:regret_rate_increase}
        R_n^{(\mathbf{b})} = \mathcal{O}\left(k\sqrt{(ms+1)n\log^3(n)}\right) \\ + \mathcal{O}\left(knc^s\sqrt{(ms+1)\log^3(n)}\right) \\ + \mathcal{O}\left(kn c^s \sqrt{\log(n)}\right) + \mathcal{O}\left(k\sqrt{n}\right) . 
    \end{multline} 
    where $c^s$ is the magnitude of the term $\left\Vert \left(\Gamma - \Gamma K C_\theta\right)^s\right\Vert_2$. 
\end{theorem}

As an improvement in the design of $u_a$ for lowering the regret bound derived in Theorem \ref{theorem:regret}, we provide Corollary \ref{corollary:regret_bound}. 

\begin{corollary}\label{corollary:regret_bound}
    Assume that $n > 0$ is finite. Let actions $a \in [k]$ be selected based on \eqref{eq:perturbation_method} where 
    \begin{equation}\label{eq:u_a_e}
        u_a^{(\mathbf{e})} = \mathbf{e}_a^t\left(\delta,s\right)\sqrt{\Theta_t^\top V_a^t\left(s\right)^{-1} \Theta_t} ,
    \end{equation}
    where $\mathbf{e}_a^t\left(s\right)$ is defined as
    \begin{multline}
        \mathbf{e}_a^t \left(\delta,s\right) \triangleq  \sqrt{2B_R^2\log\left(\frac{1}{\delta}\frac{ \det(V_a^t\left(s\right))^{1/2}}{\det\left(\lambda I_{ms+1}\right)^{1/2}}\right)} \\ 
        + \lambda B_G \sqrt{\mbox{tr}\left(V_a^t\left(s\right)^{-1}\right)} \label{eq:e_def}.
    \end{multline}
    
    Regret $\sum_{t=1}^n X_t^* - X_t$ has the following has the following asymptotic rate that is satisfied with a probability of at least $1- 9\delta$: 
    \begin{multline}\label{eq:regret_rate_increase_no_bias}
        R_n^{(\mathbf{e})} = \mathcal{O}\left(k\sqrt{(ms+1)n\log^3(n)}\right) \\ + \mathcal{O}\left(knc^s\sqrt{(ms+1)\log^3(n)}\right)+ \mathcal{O}\left(kn c^s\sqrt{\log(n)}\right) \\ + \mathcal{O}\left(n^{3/2}c^s\sqrt{(ms+1)\log^2(n)}\right) + \mathcal{O}\left(k\sqrt{n}\right) . 
    \end{multline} 
    
    % Finally, if $s$ is set such that 
    % \begin{equation}
    %     B_s\left(n\right)B_{\beta}\left(s\right) \leq \mathbf{B}_{\beta} , \label{eq:s_satisfactory}
    % \end{equation}
    % then regret has the following upper bound which is satisfied with a probability of at least $1-9\delta$
    % \begin{multline}\label{eq:regret_bound_OFU_s}
    %     R_n^{(\mathbf{e})} \leq k\sqrt{n}\overline{B}_{\theta}^\delta\left(n\right)\overline{B}_u^\delta\left(n\right) \left(\overline{B}_{(\mathbf{b})}^\delta\left(n\right) + \overline{B}_{(\mathbf{e})}^\delta\left(n\right)\right)\\ + \sqrt{2n\left(ms+1\right)\log\left(n/\delta\right)}\mathbf{B}_{\beta} \\ + 2 kn B_{\beta}\left(s\right)\sqrt{2\log\left(n/\delta\right)} + k\sqrt{4 n B_R^2 \log\left(1/\delta\right)}. 
    % \end{multline}
\end{corollary}
For interpreting the inequalities used in Theorem \ref{theorem:regret} and Corollary \ref{corollary:regret_bound}, we can view the bounds as a sum of identified Kalman filter predictor errors, which according to \cite{tsiamis2019finite} has a model error of $\mathcal{O}\left(\sqrt{(ms+1)\log\left(n\right)/n}\right) + \mathcal{O}\left(1\right)$, where $\mathcal{O}\left(1\right)$ is based on the the impacts of the bias $\beta_a^t$ on model error. The sum of model errors for one action $a \in [k]$ is therefore  $\mathcal{O}\left(\sqrt{(ms+1)n\log^3\left(n\right)}\right) + \mathcal{O}\left(n\sqrt{(ms+1)\log^3\left(n\right)}\right)$. In addition, Theorem \ref{theorem:regret} and Corollary \ref{corollary:regret_bound} imply that if dedicated exploration/exploitation phases are used, i.e. an exploration phase where each action $a \in [k]$ is sampled for $\tau$ rounds and an exploitation phase for the remaining $n-k\tau$ rounds, then $\tau$ must be tuned such that the linear increase in regret from the exploration phase does not cause the regret to be higher than the regret rates \eqref{eq:regret_rate_increase} or \eqref{eq:regret_rate_increase_no_bias}. However, if the perturbation signal $u_a$ is used, then tuning the terms $\delta \in (0,1)$ and $B_R^2 > 0$ in $u_a$ only increases the rate by a constant. 

In Corollary \ref{corollary:regret_bound}, this presents a key issue with setting $u_a$ based on the model error bound of $G_a\left(s\right)-\hat{G}_a^t\left(s\right)$. In both Theorem \ref{theorem:regret} and Corollary \ref{corollary:regret_bound}, model error impacted by the bias both contribute to superlinear regret rates. In Theorem \ref{theorem:regret}, the constant in the superlinear regret rate is impacted by the bias $\beta_a^t$ bound $\mathbf{B}_{\beta}$ which must satisfy $B_{\beta}\left(s\right) \leq \mathbf{B}_{\beta}$. Therefore, using a \textit{known} bound $\mathbf{B}_{\beta}$ can be significantly large or nonexistent if the state matrix $\Gamma$ has eigenvalues on the unit circle (the term $\sqrt{\mbox{tr}\left(Z_t\right)}$ is unbounded). However, in Corollary \ref{corollary:regret_bound}, the constant in the superlinear regret rate is impacted by the chosen parameter $s$ value, which is a parameter that only has the restriction $s > 0$. Finally, Corollary \ref{corollary:regret_bound} states that controlling $s$ can lower regret significantly. Therefore, based on our analysis, we propose to design $u_a$ such that it only considers model error from Gaussian noise $\varepsilon_a^t$, i.e. $u_a$ is set to 
\begin{equation}
    u_a = \mathbf{e}_a^t\left(\delta,s\right)\sqrt{\Theta_t^\top V_a^t\left(s\right)^{-1} \Theta_t} \label{eq:u_a_design} , 
\end{equation}
where $\mathbf{e}_a^t\left(\delta,s\right)$ is defined in \eqref{eq:e_def}. In the next subsection, since the magnitude of $\beta_a^t$ is directly impacted by the window size $s$ where the magnitude of $\beta_a^t$ decreases as $s$ increases, we will use an adaptive window size algorithm for tuning $s$ to minimize the bias $\beta_a^t$ impact on model error and prediction error. This strategy will be used instead of applying the perturbation signal $u_a$ for controlling the bias $\beta_a^t$ impact on model error. 

\subsection{Adaptive Window-Size}\label{sec:window_size}

As mentioned in Remark \ref{remark:MagnitudeofBeta}, the window-size parameter $s$ impacts the magnitude of the bias term $\beta_a^t$ . Also, according to Theorem \ref{theorem:model_error} and Lemma \ref{lemma:optimal_value}, the window-size parameter $s$ impacts the model error $\left\Vert \hat{G}_a^t\left(s\right) - G_{a}\left(s\right) \right\Vert_{V_a^t\left(s\right)}$ (and consequently, the error $\hat{G}_a^t\left(s\right)^\top \Theta_t - G_{a}\left(s\right)^\top \Theta_t$). The following theorem provides a method to adaptively control the parameter $s$ so as to minimize the impact of the bias term and the model error on the model $\hat{G}_a^t\left(s\right)$ that predicts the reward $X_t$. Proof for Theorem \ref{theorem:s_cost} is found in Appendix \ref{appendix:s_cost}.

\begin{theorem}\label{theorem:s_cost}
Let $J_a^t\left(s\right)$  be a cost function defined as:
\begin{multline}\label{eq:s_param_cost_function}
J_a^t\left(s\right) \triangleq \left\vert X_t - \hat{G}_a^t\left( s\right)^\top  \Theta_{t} \right\vert \\+ \nu \mathbf{e}_a^t \left(\delta,s\right) \sqrt{\Theta_{t}^\top V_a^t\left(s\right)^{-1} \Theta_{t}} . 
\end{multline}

Then, there exists $\delta\in(0,1)$ such that the following inequality is satisfied with a probability of at least  $1-\delta$:
\begin{multline}\label{eq:upper_bound_2}
J_a^t\left(s\right) +\left\vert\varepsilon_a^t\right\vert \geq \\ \Bigg\vert \left(\left(1-\nu\right)\left(\lambda G_a\left(s\right) - \mathbf{E}_{\mathscr{T}_a}\mathbf{O}_{\mathscr{T}_a}^\top\right)\right) V_a^t\left(s\right)^{-1}\Theta_t\\ - \mathbf{B}_{\mathscr{T}_a} \mathbf{O}_{\mathscr{T}_a}^\top V_a^t\left(s\right)^{-1}\Theta_t + \beta_a^t \Bigg\vert. 
\end{multline}
\end{theorem}

\begin{remark}
According to the fact that $\varepsilon_a^t$ is invariant to $s$, Theorem \ref{theorem:s_cost} suggests that minimizing the cost function $J_a^t\left(s\right)$ would minimize the bias term $\beta_a^t$. 
\end{remark}

\begin{remark}
The difference $\left\vert X_t - \hat{G}_a^t\left(s\right)^\top \Theta_t\right\vert$ in the cost function $J_a^t\left(s\right)$ of Theorem \ref{theorem:s_cost} is noisy due to the reward $X_t$ and the vector of contexts $\Theta_t$. However, using an average of $\left\vert X_t - \hat{G}_a^t\left(s\right)^\top \Theta_t\right\vert$ does not consider how $\hat{G}_a^t\left(s\right)$ improves as the number of samples $N_a$ increases. Therefore, we optimize for the window size $s$ using the following cost function 
\begin{equation}
    \begin{array}{cc}
        & \hat{J}_a^t\left(s\right) = \zeta_a^t\left(s\right) + \nu \mathbf{e}_a^t\left(\delta,s\right)\sqrt{\Theta_t^\top V_a^t\left(s\right) \Theta_t}\\
        \mbox{ s.t. } & \zeta_a^{t+1}\left(s\right) = \alpha\zeta_a^t\left(s\right) + \left(1-\alpha\right)\left\vert X_t - \hat{G}_a^t\left(s\right)^\top \Theta_t\right\vert
    \end{array}\nonumber. 
\end{equation}
\end{remark}

\section{Bandit Strategy}\label{sec:Bandit_Strategy}
The strategy the learner will use when choosing an action $a \in [k]$ is based on the theorems and lemmas mentioned above which can be broken down into two parts: \textbf{Part 1:} selecting parameter $s$ and \textbf{Part 2:} selecting action $a \in [k]$. This is comprehensively shown in Algorithm \aresref{}. Fig. \ref{figure:algorithm} illustrates \textbf{Part 1} and \textbf{Part 2} of the algorithm.

\textbf{Part 1:} A model $\hat{G}_a^t\left(s\right) \gets \mathbf{0}_{(ms+1) \times 1}$ is set for each parameter $s \in \mathscr{S} \triangleq \{0,1,\dots,s_N\}$ and action $a \in [k]$, where $s_N$ is a preset value for the maximum window size of previous observations that are used. This provides $s_N  \times k$ models. For each parameter $s \in \mathscr{S} \triangleq \{0,1,\dots,s_N\}$ and action $a \in [k]$, the cost function $\hat{J}_a^t\left(s\right)$ is computed. Model $\hat{G}_a^t\left(s\right)$ uses parameter $s_a$ that minimizes $\hat{J}_a^t\left(s\right)$, i.e. $s_a = \arg\min_{s \in \{0,1,\dots,s_N\}}\hat{J}_a^t\left(s\right)$. However, for all the perturbation terms $u_a$, $a \in [k]$, they will use the maximum $s_a$, i.e. $s = \max_{a\in[k]} s_a$. 

\textbf{Part 2:} The action $a \in [k]$ that the learner chooses is $a = \arg\max_{a \in [k]} \hat{G}_a^t\left(s\right)^\top \Theta_t + u_a$ , where $u_a$ is defined in \eqref{eq:u_a_design}. 

\begin{remark}
    Note that in lines 19 and 20 of Algorithm \aresref{} implies that the $u_a$ is set with the maximum $s_a$ out of all the actions. Not implementing lines 19 and 20 will cause the algorithm to commit to the action with the highest $s_a$ value instead of exploring all the actions. 
\end{remark}

\begin{center}
\textbf{Algorithm: \\ Adaptive Recursive least-squares Exploration of System (ARES)}
\end{center}

\textbf{Input}: $\delta, \alpha \in (0,1)$, $s_N, \lambda, B_R > 0$ 
\begin{enumerate}[leftmargin=*, label=Step \arabic*:]
    \item Initialize: $t \gets 1$, $\mathscr{S} = \{0,1,\dots,s_N\}$
    \item For each $a \in [k]$:
    \begin{enumerate}[leftmargin=*, label=(\alph*)]
        \item $\mathscr{T}_a \gets \{\}$
        \item For each $s \in \{0,1,\dots,s_N\}$:
        \begin{enumerate}[leftmargin=*, label=(\roman*)]
            \item $M_a \gets 0$
            \item $N_a \gets 0$
            \item $V_a \gets \lambda I_{ms+1 \times ms+1}$
            \item $\hat{G}_a^t(s) \gets \mathbf{0}_{(ms+1) \times 1}$
            \item $\mathbf{e}_a^t(\delta,s) \gets 1/\epsilon$ (where $\epsilon$ is small)
            \item $\hat{J}_a^t(s) \gets 0$
        \end{enumerate}
    \end{enumerate}
    \item For each $t \in [n]$:
    \begin{enumerate}[leftmargin=*, label=(\alph*)]
        \item For each $a \in [k]$:
        \begin{enumerate}[leftmargin=*, label=(\roman*)]
            \item $s_a \gets \arg\min_{s \in \{0,1,\dots,s_N\}} \hat{J}_a^t(s)$
        \end{enumerate}
        \item $s \gets \max_{a \in [k]} s_a$
        \item $u_a$ calculation:
        \begin{enumerate}[leftmargin=*, label=(\roman*)]
            \item If $t \geq s_a$: \\ $u_a \gets \mathbf{e}_a^t(\delta,s)\sqrt{\Theta_t^\top V_a^t(s)^{-1} \Theta_t}$
            \item If $t < s_a$: \\ $u_a \gets \sqrt{2\frac{\log(1/\delta)}{M_a}}$
        \end{enumerate}
        \item $a^*$ calculation:
        \begin{enumerate}[leftmargin=*, label=(\roman*)]
            \item If $t \geq s_a$: \\ $a^* \gets \underset{a \in [k]}{\arg\max} \hat{G}_a^t(s_a)^\top \Theta_t + u_a$
            \item If $t < s_a$: \\ $a^* \gets \underset{a \in [k]}{\arg\max} \sum_{t \in \mathscr{T}_a} \frac{X_t}{M_a} + u_a$
        \end{enumerate}
        \item Sample $(X_t, \theta_t)$ from Eq \eqref{eq:Linear_System}
        \item $\mathscr{T}_{a^*} \gets \mathscr{T}_{a^*} \cup \{t\}$
        \item Update counts:
        \begin{enumerate}[leftmargin=*, label=(\roman*)]
            \item If $t < s_{a^*}$: \\ $M_{a^*} \gets M_{a^*} + 1$
            \item If $t \geq s_{a^*}$: \\ $N_{a^*} \gets N_{a^*} + 1$
        \end{enumerate}
        \item For each $s \in \{0,1,\dots,\min(t,s_N)\}$:
        \begin{enumerate}[leftmargin=*, label=(\roman*)]
            \item $V_{a^*}^{t+1}(s) \gets V_{a^*}^t(s) + \Theta_\tau \Theta_\tau^\top$
            \item $\hat{G}_{a^*}^{t+1}(s)^\top \gets \hat{G}_{a^*}^t(s)^\top \\ + X_\tau \Theta_\tau^\top V_{a^*}^t(s)^{-1}$
            \item $\zeta_{a^*}^{t+1}(s) \gets \alpha\zeta_{a^*}^t(s) \\ + (1-\alpha)\left|X_t - \hat{G}_{a^*}^t(s)^\top \Theta_t\right|$
            \item $\hat{J}_{a^*}^t(s) \gets \zeta_{a^*}^t(s) \\ + \mathbf{e}_{a^*}^t(\delta,s) \sqrt{\Theta_t^\top V_{a^*}^t(s)\Theta_t}$
            \item Update $\mathbf{e}_{a^*}^t(\delta,s)$ based on Eq \eqref{eq:e_def}
        \end{enumerate}
    \end{enumerate}
\end{enumerate}

\begin{figure}[ht]
    \centering
    \includegraphics[width=\linewidth]{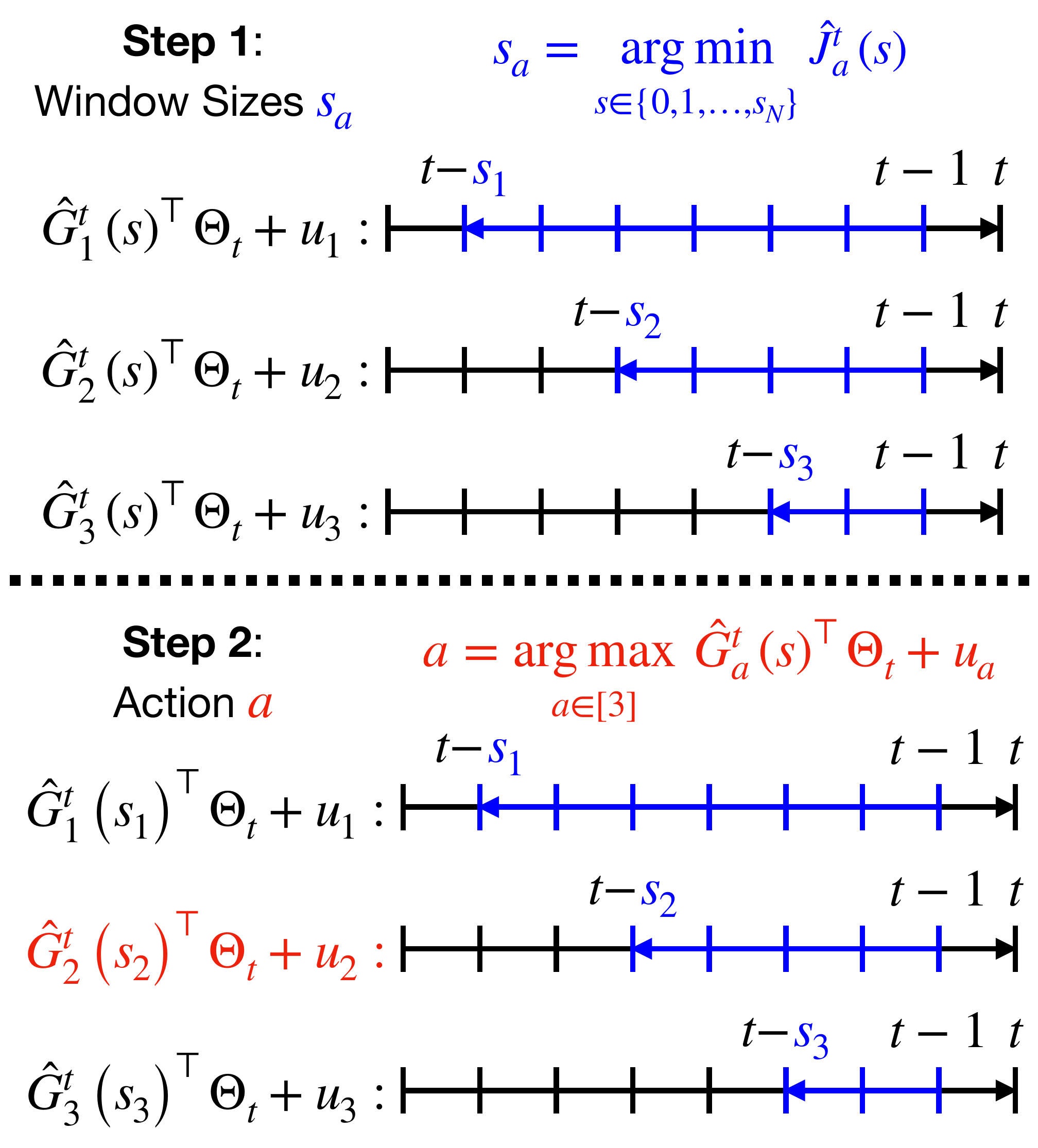}
    \caption{In step 1, the learner selects the window size $s_a$ for each action $a \in [k]$. As for step 2, the learner selects the action $a \in [k]$ that maximizes the quantity $\hat{G}_a^t\left(s_a\right)^\top\Theta_t + u_a$ with the chosen $s_a$. }
    \label{figure:algorithm}
\end{figure}

\subsection{How tight is ARES Regret Upper Bound?}

To study the tightness of the ARES regret upper bound derived in Corollary \ref{corollary:regret_bound}, we first prove the regret lower bound. The regret lower bound provides a metric of how difficult it is to consistently select the optimal action $a_t^*$. To acquire this bound, we first provide the following \textit{Oracle} which selects actions $c_a \in \mathcal{A}$ based on the following optimization problem 
\begin{equation}\label{eq:Kalman_Oracle}
    \begin{array}{cc}
        \underset{a \in [k]}{\arg\max} & \left\langle c_a, \tilde{z}_t\right\rangle \\
        \mbox{ s.t. } & \tilde{z}_{t+1} = \Gamma \tilde{z}_t + \Gamma \tilde{K} \left(Y_t - \tilde{C} \tilde{z}_t\right)
    \end{array}
\end{equation}
where $\tilde{K}$, $\tilde{C}$, and their associated terms are defined as 
\begin{align}  
    \tilde{K} & \triangleq \tilde{P} \tilde{C}^\top \left(\tilde{C} \tilde{P} \tilde{C}^\top + \tilde{R}\right)^{-1} \nonumber \\
    \tilde{P} & = \Gamma \tilde{P} \Gamma^\top + Q \nonumber \\
    & ~~~~~ - \Gamma \tilde{P} \tilde{C}^\top \left(\tilde{C} \tilde{P} \tilde{C}^\top + \tilde{R}\right)^{-1}\tilde{C} \tilde{P}\Gamma^\top \nonumber \\
    \tilde{C} & \triangleq \begin{pmatrix}
        c_1^\top &
        c_2^\top &
        \dots &
        c_k^\top &
        C_\theta^\top
    \end{pmatrix}^\top \nonumber \\
    \tilde{R} & \triangleq \mbox{blockdiag}\left(\sigma^2, \dots,\sigma^2, R\right) \nonumber \\
    Y_t & \triangleq \Tilde{C} z_t + \nu_t, ~ \nu_t \sim \mathcal{N}\left(\mathbf{0},\tilde{R}\right) \nonumber. 
\end{align}

The \textit{Oracle} is the method that selects actions $c_a \in \mathcal{A}$ using the context $\theta_t$ and rewards $X_t$ for \textit{each} corresponding action $c_a\in \mathcal{A}$. In other words, this method is not viable with our assumptions as it observes more information than the problem formulation allows. Using the \textit{Oracle} action selection method in \eqref{eq:Kalman_Oracle}, the following theorem proves a \textit{lower bound} on its lowest obtainable regret. Proof for Theorem \ref{theorem:lower_bound_discrete} is found in Appendix \ref{appendix:lower_bound}. 

\begin{theorem}\label{theorem:lower_bound_discrete}
    Let there regret $R_n$ \eqref{eq:regret}. The lower bound for regret $R_n$ is the following inequality 
    \begin{equation}\label{eq:regret_lower_bound_discrete}
        R_n \geq n \sum_{i \in [k]} \sum_{j \in [k]} \sqrt{\frac{2\left(c_{a_j} - c_{a_i} \right)^\top Z \left(c_{a_j} - c_{a_i}\right)}{\mbox{tr}\left(\Psi_{i|j}\right)^{2k-2} \left\vert \Tilde{\Sigma}_{i|j}\right\vert }}
    \end{equation}
    where $\Tilde{\Sigma}_{i|j},\Psi_{i|j}$ are defined to be 
    \begin{align}
        \Tilde{\Sigma}_{i|j} & \triangleq A_i \Tilde{Z} A_i^\top - A_i\Tilde{Z} A_j^\top \left(A_j Z A_j^\top \right)^{-1} A_j \Tilde{Z} A_i^\top \label{eq:tilde_sigma_ij}\\
        \Psi_{i|j} & \triangleq \begin{pmatrix} \Sigma_{i|j}^{-1} & \Sigma_{i|j}^{-1}\Pi_{i|j} \\             
        \Pi_{i|j}^\top \Sigma_{i|j}^{-1} & \Pi_{i|j}^\top \Sigma_{i|j}^{-1}\Pi_{i|j}         
        \end{pmatrix}\label{eq:psi_ij}       
    \end{align}
    which are based on the following defined terms
    \begin{align}
        \Tilde{Z} & \triangleq Z-\tilde{P} \\
        \begin{pmatrix}
            A_i \left(z_t - e_{t|t-1}\right) \\
            A_j z_t
        \end{pmatrix} & \sim \mathcal{N}\left(\mathbf{0}, \Sigma_{i,j}\right) \nonumber \\
        \Sigma_{i,j} & \triangleq \begin{pmatrix}
            A_i \Tilde{Z} A_i^\top & A_i\Tilde{Z} A_j^\top \\
            A_j \Tilde{Z} A_i^\top & A_j Z A_j^\top 
        \end{pmatrix} \label{eq:oracle_covariance} \\
        A_i \left(z_t - e_{t|t-1}\right) \mid A_j z_t & \sim \mathcal{N}\left( \Pi_{i|j}z_t , \Tilde{\Sigma}_{i|j}\right) \nonumber \\
        A_i & \triangleq \begin{pmatrix}
            c_{a_i} - c_{a_1'} & \dots & c_{a_i} - c_{a_{k-1}'}
        \end{pmatrix}^\top \nonumber  \\
        \Pi_{i|j} & \triangleq A_i\Tilde{Z} A_j^\top \left(A_j Z A_j^\top \right)^{-1} A_j  \nonumber, 
    \end{align}
    and $Z$ solves the Lyapunov equation $Z = \Gamma Z \Gamma^\top + Q $. 
\end{theorem}

From Theorem \ref{theorem:lower_bound_discrete}, the \textit{Oracle} cannot obtain a regret rate lower than linear. This follows from what was proven in Theorem 1 in \cite{besbes2014stochastic} where the authors have proven that no bandit algorithm can obtain a regret lower than linear in environments that linearly change. Therefore, based on the asymptotic rate obtained in Corollary \ref{corollary:regret_bound}, having a superlinear regret rate of at most $n^{3/2} c^s \log\left(n\right)$ is appropriate, which can be lowered based on the design of window-size $s$.

\section{Evaluation Study}\label{sec:Numerical_Simulation}

The proposed Algorithm \aresref{} is compared with an Oracle (where the chosen action $a \in [k]$ is the action $a$ that has the maximum $\hat{X}_t$ from the Kalman filter \eqref{eq:Kalman_Filter}), UCB \cite{agrawal1995sample}, Sliding Window UCB \cite{garivier2008upper}, Rexp3 \cite{besbes2014stochastic}, Random (for each round, the action is chosen uniformly random), and PIES \cite{9993316}. PIES is similar to ARES where in the first $ks$ rounds the algorithm sequentially chooses each action $a \in [k]$ to identify the each vector $G_a\left(s\right)$. After the $ks$ rounds, PIES then chooses the action based on what it predicts to return the highest reward. In addition, PIES is preset with a parameter $s \in \{1,2,\dots,9,10\}$ which is the fixed window size of the algorithm. Each action has the same $s \in \{1,2,\dots,9,10\}$ value. The parameters used for ARES are in Table \ref{table:parameters}. 

% {
% \footnotesize
% \begin{table}[ht]

%     \centering
%     \begin{center}
%         \caption{ARES Defined Parameters}
%     \end{center}
    
%     \begin{tabular}{|c|c|}
%         \hline
%         \textbf{Parameter} & \textbf{Value} \\
%         \hline
%         $\delta$: & 0.1 \\
%         $\lambda$: & 1 \\
%         $\alpha$: & 0.99 \\
%         $s_N$: & 10 \\
%         $\nu$: & 0.1 \\
%         $B_R$: & $\max_a c_a^\top P c_a + \sigma_\eta^2$ \\
%         $B_G$: & $100$ \\
%         $B_c$: & $1$ \\
%         \hline
%     \end{tabular}
%     \label{table:parameters}
% \end{table}
% }

\begin{table}[ht]

    \centering
    \begin{center}
        \caption{ARES Defined Parameters}\label{table:parameters}
    \end{center}
    \resizebox{\columnwidth}{!}{
    \begin{tabular}{|c|c|c|c|c|c|c|c|}
        % \textbf{Parameter} & \textbf{Value} \\
        \hline
        $\delta$ & $\lambda$ & $\alpha$ & $s_N$ & $\nu$ & $B_R$ & $B_G$ & $B_c$ \\
        \hline
        0.1 & 1 & 0.99 & 10 & 0.1 & $\max_a c_a^\top P c_a + \sigma_\eta^2$ & 100 & 1 \\
        \hline
    \end{tabular}
    }
\end{table}

\begin{remark}
    The Oracle chooses actions based on the optimal prediction of the rewards from \eqref{eq:Linear_System} in the mean-squared error sense as it assumes knowledge of \eqref{eq:Linear_System}. Therefore, the Oracle is not applicable for the considered setting. Rather, we use the Oracle as a baseline for regret performance.
\end{remark}

For the numerical example, we consider the following LGDS
\begin{equation}\label{eq:LGDS_Good}
    \begin{cases}
        z_{t+1} & = \Psi z_t + \xi_t \\
        \theta_t & = C_\theta z_t + \phi_t\\
        X_t & = \left\langle \mathbf{e}_a, z_t \right\rangle + \eta_t
    \end{cases}, 
\end{equation}
\begin{equation}\label{eq:noise_stats}
    \xi_t \sim \mathcal{N}\left(0,I_{5}\right), ~\phi_t \sim \mathcal{N}\left(0,1\right), ~ \eta_t \sim \mathcal{N}\left(0,1\right) \nonumber \nonumber,
\end{equation}
where $\Psi \in \mathbb{R}^{5 \times 5}$ is defined as follows
\begin{align}
    \Psi & = \frac{0.99}{\rho\left(T\right)}T, ~ T[i,j] = \begin{cases}
        2^{- \psi \left\vert i - j \right\vert } & \mbox{ if } i < j \\
        - 2^{- \psi \left\vert i - j \right\vert } & \mbox{ if } i \geq j 
    \end{cases} \nonumber \\
    C_\theta & \triangleq \begin{pmatrix}
        1 & 0 & 0 & 0 & 0
    \end{pmatrix} \nonumber, 
\end{align}
\begin{equation}
    \mathbf{e}_a \in \mathcal{A} \triangleq \left\{\begin{pmatrix}
        0 & 0 & 0 & 1 & 0
    \end{pmatrix}^\top, \begin{pmatrix}
        0 & 0 & 0 & 0 & 1
    \end{pmatrix}^\top\right\} \nonumber. 
\end{equation}

The LGDS \eqref{eq:LGDS_Good} is computed for $10^4$ rounds to set the state variable $z_t$ to a steady-state distribution. The total number of interactions between the algorithm and the environment is $n = 10^4$. Each algorithm is average across 100 different simulations.

The rationale behind \eqref{eq:LGDS_Good} is as follows. ARES is designed such that it is adaptively exploring using $u_a$ and adjusting window size $s$ (number of contexts $\theta_{t-s}, \dots, \theta_{t-s}$ in the past used to predict each action's reward $X_t$). Recall that the magnitude of $u_a$ is impacted by the bound $B_R^2 \geq \mbox{Var}\left(\varepsilon_a^t\right)$ where larger $u_a$ values lead to longer exploration times. For the bias term $\beta_a^t$, note that the magnitude of the bias term decreases exponentially as the parameter $s$ increases where the rate of decrease is dependent on the spectral radius $\rho\left(\Psi - \Psi K C_\theta\right)$. Therefore, we design the LGDS \eqref{eq:LGDS_Good} parameterized by $\psi$ since $\psi$ in the LGDS \eqref{eq:LGDS_Good} impacts both $\mbox{Var}\left(\varepsilon_a^t\right)$ and $\rho\left(\Psi - \Psi K C_\theta\right)$ where $K$ is based on the steady-state Kalman filter \eqref{eq:ss_Kalman_Filter}. For the simulations, parameter $\psi$ is set to $\psi = 1,2,3,4$. 

In Figure \ref{figure:regret}, the plots on the left are the median percentage of regret decreased when using ARES over the comparison algorithm. For the plots on the right, each line is each algorithm's regret value normalized  with respect to the Oracle's regret. Note that we only plot the PIES's normalized regret with parameters $s=1$ and $s=10$ since parameter $s = 1$ is where PIES has the highest regret value while parameter $s = 10$ is where PIES has the lowest regret value. It appears that both PIES and ARES perform better than any of the other comparison algorithms. For $\psi=1,2,3$, there exists a parameter $s$ value such that PIES does better than ARES. Each subtitle contains the set $\psi$ value, the maximum $c_a^\top P c_a + \sigma_\eta^2$ value (where $P$ is the steady-state Kalman filter \eqref{eq:ss_Kalman_Filter} error covariance matrix), and the spectral radius $\rho\left(\Psi - \Psi K C_\theta\right)$. It appears that as $\max_a c_a^\top P c_a + \sigma_\eta^2$ and $\rho\left(\Psi - \Psi K C_\theta\right)$ increase, ARES starts to perform better than PIES for any preset parameter $s=1,\dots,10$ value. For $\psi = 1,3$, there exists a parameter  $s$ value such that PIES performs better than ARES, where ARES's regret is at most $10$\% more than than PIES regret. Finally, it appears that ARES and PIES normalized regrets appear to get closer to 1 (the Oracle's regret value) as $\psi$ is closer to 1. This implies that as $\psi$ increases, the environment for PIES and ARES becomes harder, i.e. the model $G_a\left(s\right)$ is more difficult to identify.

Figure \ref{figure:perturbation} analyzes if either the adaptive exploration $u_a$ term or adaptive window size selector $s$ is the cause for ARES worse performance in comparison to PIES for $\psi = 1,3$. In Figure \ref{figure:perturbation}, the plots on the left are the average perturbation $u_a$ values for each action  and the plots on the right are the average parameter $s$ values. Note that the average parameter $u_a$  values for each $\psi$ value is approximately the same. However, for the right plots, the average parameter $s$ value for each action is larger for $\psi = 1,2,3$ in comparison to the average parameter $s$ value for $\psi = 4$. In addition, we can observe that for $\psi = 1,3$ the average parameter $s$ for each action slowly increases while for $\psi=2,4$ the values stay approximately constant after round $t=2500$. The left plots of Figure \ref{figure:regret} for PIES performance at $\psi = 1,2,3$ imply that ARES needs a large $s$ value to obtain better performance. Finally, the spectral radius $\rho\left(\Psi - \Psi K C_\theta\right)$ increases as $\psi$ increases, implying that the parameter $s$ value is less impactful and therefore small $s$ values are sufficient. Therefore, Figure \ref{figure:perturbation} implies that if ARES adaptive window selector takes time to select high parameter $s$ values, then it will do worse that PIES with high preset $s$ values.

The parameter $s$ value in ARES is chosen such that it minimizes $\hat{J}_a^t\left(s\right)$ \eqref{eq:s_param_cost_function}, where $\hat{J}_a^t\left(s\right)$ consists of two terms: $\left\vert X_t - \hat{G}_a^t\left(s\right)^\top \Theta_t\right\vert$ and $\nu \mathbf{e}_a^t\left(\delta,s\right) \sqrt{\Theta_t^\top V_a^t\left(s\right) \Theta_t }$. The term $\nu \mathbf{e}_a^t\left(\delta,s\right) \sqrt{\Theta_t^\top V_a^t\left(s\right) \Theta_t }$ is similar to the perturbation input $u_a$; therefore, we only have to analyze the impact of $\left\vert X_t - \hat{G}_a^t\left(s\right)^\top \Theta_t\right\vert$. Recall that PIES identifies $G_a\left(s\right)$ and uses the identified vector predicting the reward $X_t$ for each action $a \in [k]$ given a set parameter $s$  value, which is similar to ARES. Therefore, analyzing the PIES's reward prediction error for each parameter $s$ value gives intuition on $\left\vert X_t - \hat{G}_a^t\left(s\right)^\top \Theta_t\right\vert$ . In Figure \ref{figure:parameter_s}, the left plots are the average errors for each $a = 1$ for each parameter $s = 1,2,\dots,10$. For the right plots, these are the average errors of action $a = 2$  for each parameter $s = 1,2,\dots,10$. As shown, it appears that ARES on average chooses the parameter $s$ value that minimizes the prediction error $\left\vert X_t - \hat{G}_a^t\left(s\right)^\top \Theta_t\right\vert$ the most for each action. Figure \ref{figure:parameter_s} also provides intuition on why ARES chooses a larger parameter $s$ value for one action versus another action. For example, in the third lowest plot of Figure \ref{figure:parameter_s}, action $a = 1$ error for parameters $s= 8,9,10$ are similar. However, for action $a = 2$, its error for parameter $ s= 10$ is the smallest value. Since we penalize for choosing high parameter $s$ values with $\nu \mathbf{e}_a^t\left(\delta,s\right) \sqrt{\Theta_t^\top V_a^t\left(s\right) \Theta_t }$, ARES chooses the smallest parameter $s$ value action $a = 1$ which is $s = 8$, while ARES chooses parameter $s = 10$ for action $a = 2$  since it is the value that decreases the prediction error the most. Therefore, the takeaway from Figure \ref{figure:parameter_s} is if the prediction errors are approximately the same for a subset of parameter $s$ values, ARES will choose the minimum parameter $s$ value from the subset.

\begin{figure}[ht]
    \centering
    \includegraphics[width=\linewidth]{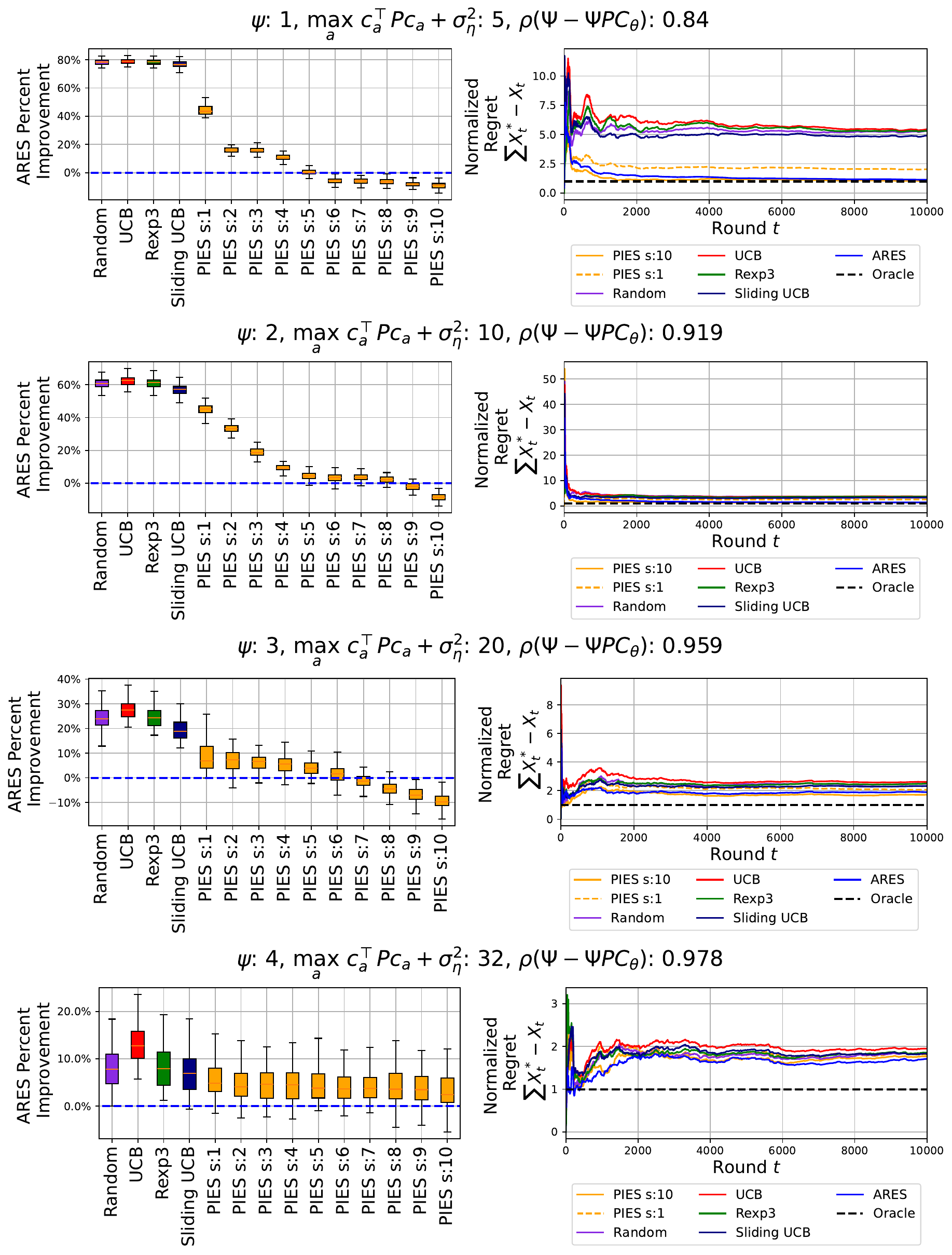}
    \caption{Plots on the left are median decrease/increase (positive percent/negative percent) in regret with respect to ARES's regret. The bottom and top part of the boxes are the first and third quantiles, respectively. Plots on the right are the median regret for each algorithm for each round. }
    \label{figure:regret}
\end{figure} 

\begin{figure}[ht]
\centering
    \includegraphics[width=\linewidth]{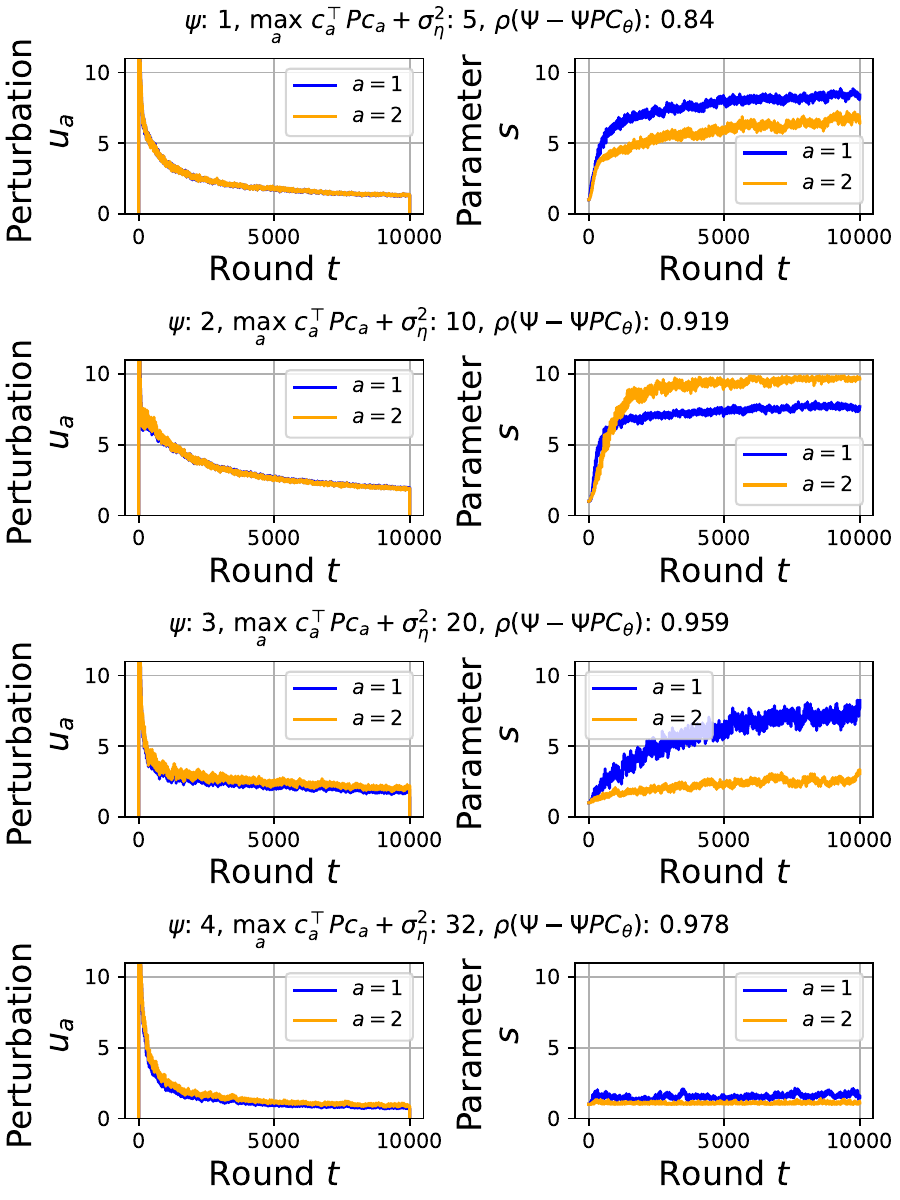}
    \caption{Plots on the left are average $u_a$ values for actions $a = 1,2$. Plots on the right are average chosen parameter $s$ values for actions $a = 1,2$. }
    \label{figure:perturbation}
\end{figure}

\begin{figure}[ht]
    \centering
    \includegraphics[width=\linewidth]{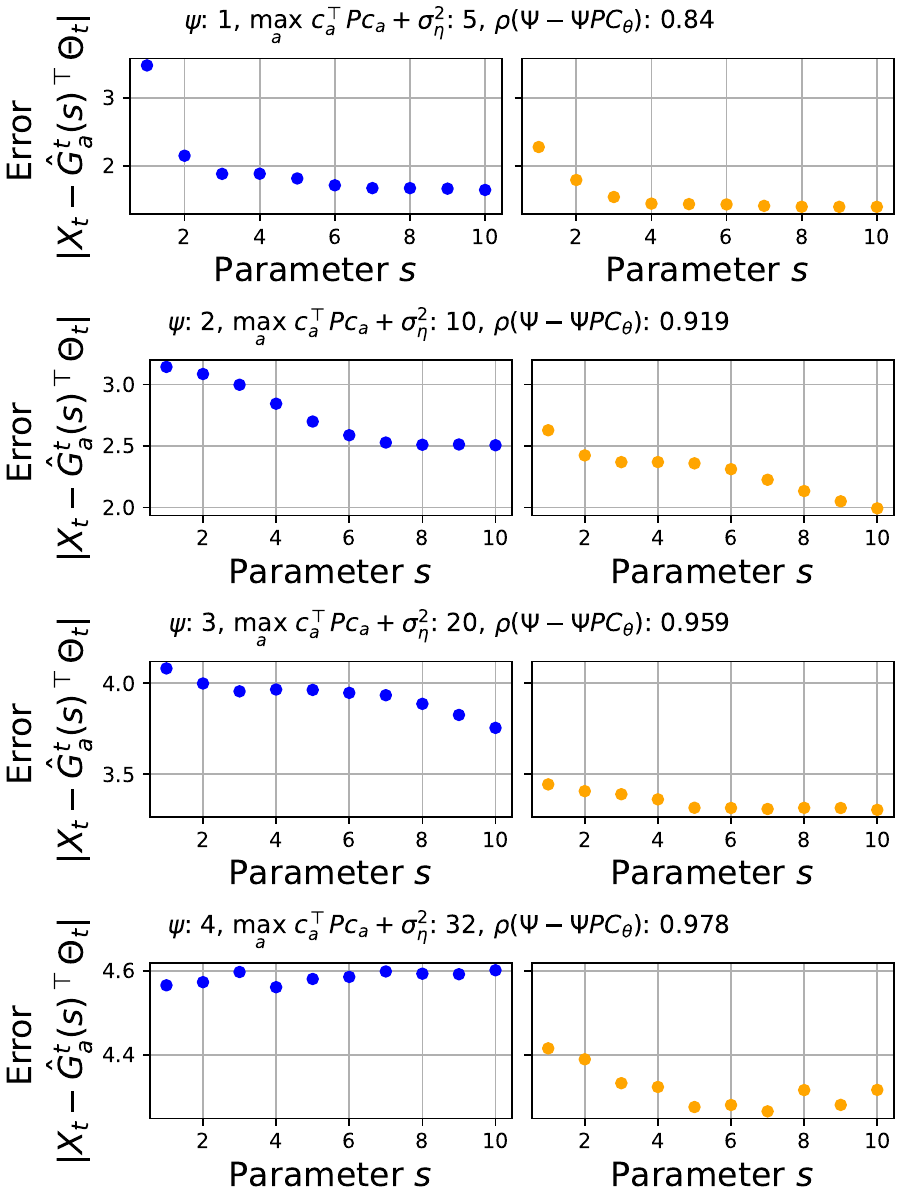}
    \caption{Plots on the left are average errors $\left\vert X_t - \hat{G}_a^t\left(s\right)^\top \Theta_t \right\vert$ for actions $a = 1$ (left plot in yellow) and $a = 2$ (right plot in blue). }
    \label{figure:parameter_s}
\end{figure}

% \subsection{How tight is ARES Regret Upper Bound?}

% To study the tightness of the ARES regret upper bound derived in Corollary \ref{corollary:regret_bound}, we plot a ratio of the derived bound \eqref{eq:regret_bound_OFU_e} (in the numerator) and the average computed regret values of ARES (in the denominator). We set $\delta = 0.1$. In the top plot of Figure \ref{figure:tightness}, the ratio of regrets are shown for $\psi = 4$ while the left plot of Figure \ref{figure:tightness} is the ratio of regrets for $\psi = 1$. Notice that in the left plot as parameter $s$ increases from $s = 1$ to $s=10$ the differences between the ratios increase. However, for the right plot, the differences between the ratios of regrets for parameters $s=1,3,10$ are smaller. Therefore, this follows from what was observed in the earlier figures where parameter $s$ has a larger impact when $\psi$ is close to $1$. In addition, we can observe that the regret bound becomes tighter as $\psi$ as parameter $s$ has a larger impact, verifying why for system \eqref{eq:LGDS_Good} with parameter $\psi = 1$ ARES chooses larger parameter $s$ values. 

% \begin{figure}[ht]
%     \centering
%     \includegraphics[width=\linewidth]{figures/tightness_figures.pdf}
%     \caption{Each  plot  is a ratio of the theoretical regret bound \eqref{eq:regret_bound_OFU_e} and the computed average  regret from the ARES simulations. }
%  \label{figure:tightness}
% \end{figure}

\section{Conclusion}\label{sec:Conclusion}
This work has provided a perspective for extracting representations from non-stationary environments to predict rewards for each action. The results in this paper provided an algorithm that adaptively explores environments with rewards generated by linear dynamical systems. This is accomplished by utilizing an adaptive exploration method inspired by the UCB algorithm and also adjusting the size of the linear predictor for predicting rewards for each action. Future work will focus on cases where the context $\theta_t$ may not be available at different time periods or just absent.

\appendices

\setupautartappendix

\section{Appendix Overview}

To help with the flow of discussion in the main text, we provided all the technical details on the algorithm ARES in the appendix. In this section, we first discuss the bounds on the model error of $\hat{G}_a^t\left(s\right)$ and the prediction error of $\hat{G}_a^t\left(s\right)^\top \Theta_t$. Next, we provide a bound on the regret if actions are selected according to \eqref{eq:perturbation_method}. This provides a methodology for designing $u_a \in \mathbb{R}$ by utilizing this regret bound. Finally, to quantify how well action-selection method \eqref{eq:perturbation_method} performs with respect to the tuned $u_a \in \mathbb{R}$ from the provided analysis, we provide a lower bound, which is a measure of how difficult it is to consistently select the optimal action $c_a \in \mathcal{A}$.

\section{STEP 1 and STEP 2 for Designing the Perturbation}\label{appendix:step_1_and_step_2}

For this section, we will provide the technical details for addressing \textbf{STEP 1} and \textbf{STEP 2}. Recall that \textbf{STEP 1} is deriving a bound on the model error of $\hat{G}_a^t\left(s\right)$ while \textbf{STEP 2} is deriving a bound on the prediction error of $\hat{G}_a^t\left(s\right)^\top \Theta_t$. 

\subsection{STEP 1}\label{appendix:step_1}

To provide a bound on the model error of $\hat{G}_a^t\left(s\right)$, first we need to provide a bound on $\left\Vert\mathbf{B}_{\mathscr{T}_a}\right\Vert_2$. The lemma below will be used for proving the bound of $\left\Vert\mathbf{B}_{\mathscr{T}_a}\right\Vert_2$. 

\begin{lemma}\label{lemma:subgaussian_half-normal}
    Let $X \in \mathbb{R}^k$ be a normally distributed random variable, i.e. $X \sim \mathcal{N}\left(0,\Sigma\right)$. Define $\tilde{c} > 0$ to be an absolute constant and $K = \sqrt{\left(1/2\right)\log\left(2\right)}$. The random variable $X^\top X$ bounded by the function $g_{\Sigma}\left(\delta\right)$ with a probability of at least $1-\delta$, 
    \begin{multline}\label{eq:g_sigma_def}
        g_{\Sigma}\left(\delta\right) \triangleq \mbox{tr}\left(\Sigma \right) + \\ \max \left\{\sqrt{\frac{K^4\left\Vert \Sigma \right\Vert_{HS}^2}{\tilde{c}}\log\left(\frac{2}{\delta}\right)}, 
            \frac{K^2\left\Vert \Sigma \right\Vert_2}{\tilde{c}}\log\left(\frac{2}{\delta}\right)\right\}. 
    \end{multline}
\end{lemma}

\begin{proof}
    First, we will prove that $X^\top X$ is bounded by a half-normally distributed random variable. Note the following equations
    \begin{equation*}
        X^\top X = \left( \Sigma^{1/2} Y\right)^\top \left( \Sigma^{1/2} Y\right), 
    \end{equation*}
    \begin{equation*}
        \Rightarrow X^\top X = Y^\top \Sigma Y, 
    \end{equation*}
    where $Y \sim \mathcal{N}\left(0,I_k \right)$ and $\Lambda \triangleq \mbox{diag}\left(\lambda_1,\dots,\lambda_k\right)$ is the diagonal matrix with the eigenvalues $\lambda_1,\dots,\lambda_k$ associated with the eigenvectors $V = \left(v_1,\dots,v_k\right)$ of the covariance matrix $\Sigma$. To bound the norm, we will use the Hanson-Wright inequality \cite{vershynin2018high}. First, we neeed to find the following constant $K \geq 0$ such that 
    \begin{equation*}
        K \triangleq \max_i\inf \left\{ \epsilon \mid \mathbb{E}\left[\exp\left(Y_i^2/\epsilon^2\right)\right] \leq 2 \right\}, 
    \end{equation*}
    where
    \begin{equation*}
        \mathbb{E}\left[\exp\left(Y_i^2/\epsilon^2\right)\right] = \exp\left(\frac{1}{2\epsilon^2 }\right)
    \end{equation*}
    Therefore, to find $K$, we have the following equations
    \begin{align}
        \mathbb{E}\left[\exp\left(Y_i^2/\epsilon^2\right)\right] & \leq 2 \nonumber\\
        \exp\left(\frac{1}{2\epsilon^2 }\right) & \leq 2 \nonumber\\
        \frac{1}{2\epsilon^2 } & \leq \log\left(2\right) \nonumber\\
        2\epsilon^2 & \geq \log\left(2\right)^{-1} \nonumber, 
    \end{align}
    \begin{equation*}
        \Rightarrow \epsilon \geq \sqrt{\frac{1}{2\log\left(2\right)}} \nonumber. 
    \end{equation*}    

    Therefore, $K = \sqrt{\left(1/2\right)\log\left(2\right)}$. Next, the Hanson-Wright inequality is the following 
    \begin{multline}
        P\left(Y^\top \Sigma Y \geq \mathbb{E}\left[Y^\top \Sigma Y\right] + \epsilon\right) \leq \\ 2\exp\left(-\tilde{c}\min\left\{\frac{\epsilon^2}{K^4\left\Vert \Sigma \right\Vert_{HS}^2},\frac{\epsilon}{K^2\left\Vert \Sigma \right\Vert_2}\right\}\right), 
    \end{multline}
    where $\tilde{c} > 0$ is an absolute constant. Now let $\delta$ be set such that 
    \begin{equation*}
        \delta = 2\exp\left(-\tilde{c}\min\left\{\frac{\epsilon^2}{K^4\left\Vert \Sigma \right\Vert_{HS}^2},\frac{\epsilon}{K^2\left\Vert \Sigma \right\Vert_2}\right\}\right), 
    \end{equation*}
    \begin{align}
        \delta & = 2\exp\left(-\tilde{c}\min\left\{\frac{\epsilon^2}{K^4\left\Vert \Sigma \right\Vert_{HS}},\frac{\epsilon}{K^2\left\Vert \Sigma \right\Vert_2}\right\}\right)\nonumber \\
        \frac{1}{\tilde{c}}\log\left(\frac{2}{\delta}\right) & = \min\left\{\frac{\epsilon^2}{K^4\left\Vert \Sigma \right\Vert_{HS}^2},\frac{\epsilon}{K^2\left\Vert \Sigma \right\Vert_2}\right\}\nonumber, 
    \end{align}
    \begin{multline}
        \Rightarrow \epsilon = \\ \max \left\{\sqrt{\frac{K^4\left\Vert \Sigma \right\Vert_{HS}^2}{\tilde{c}}\log\left(\frac{2}{\delta}\right)}, 
            \frac{K^2\left\Vert \Sigma \right\Vert_2}{\tilde{c}}\log\left(\frac{2}{\delta}\right)\right\}. 
    \end{multline}

    Therefore, with a probability of at least $1-\delta$, the random variable $X^\top X$ is upper bounded by $g_{\Sigma}\left(\delta\right)$ defined in \eqref{eq:g_sigma_def}. 
    
\end{proof}

We can now prove the bound for the bias $\left\Vert\mathbf{B}_{\mathscr{T}_a}\right\Vert_2$. 

\begin{lemma}\label{lemma:bias_error}
Let $\hat{G}_a^t\left(s\right)$ and $G_{a}\left(s\right)$ be as in \eqref{eq:G_a} and \eqref{eq:identify_2}, respectively. Then, there exists $\delta\in(0,1)$ such that $\mathbf{O}_{\mathscr{T}_a}\mathbf{B}_{\mathscr{T}_a}^\top$ satisfies the following inequality with a probability of at least $1-\delta$: 
\begin{multline}\label{eq:bias_bound_matrix}
    \left\Vert\mathbf{O}_{\mathscr{T}_a}\mathbf{B}_{\mathscr{T}_a}^\top \right\Vert_{V_a^t\left(s\right)^{-1}} \leq \\ \sqrt{N_a g_{\Sigma_{z_t}}\left(\delta/N_a\right)}B_{\beta}\left(s\right)\sqrt{\mbox{tr}\left(I_{ms+1} - \lambda V_a^t\left(s\right)^{-1}\right)} , 
\end{multline} 
where $B_{\beta}\left(s\right)$ is defined to be \eqref{eq:B_beta_def} and $B_c$ is as in Assumption \ref{assum:action_bound}.

\end{lemma}

\begin{proof}
According to \eqref{eq:matrix_definitions}, we have:
\begin{align}\label{eq:B_T_norm}        \left\Vert\mathbf{O}_{\mathscr{T}_a}\mathbf{B}_{\mathscr{T}_a}^\top \right\Vert_{V_a^t\left(s\right)^{-1}} = &\left\Vert \sum_{t \in \mathscr{T}_a} \Theta_t \beta_a^t \right\Vert_{V_a^t\left(s\right)^{-1}}  \nonumber\\
=&\left\Vert \sum_{t \in \mathscr{T}_a} V_a^t \left( s\right)^{-1/2}\Theta_t \beta_a^t \right\Vert_2,
\end{align}
which according to the triangle inequality and Cauchy-Schwarz inequality implies that:
 \begin{equation}
\left\Vert\mathbf{O}_{\mathscr{T}_a}\mathbf{B}_{\mathscr{T}_a}^\top \right\Vert_{V_a^t\left(s\right)^{-1}} \leq \sum_{t \in \mathscr{T}_a}\left\vert\beta_a^t \right\vert\sqrt{\Theta_{t}^\top V_a^t\left(s\right)^{-1}\Theta_{t}} \label{eq:B_T_norm_2}. 
\end{equation} 

According to \eqref{eq:residual_error_beta}, Assumption \ref{assum:action_bound}, and Remark \ref{remark:MagnitudeofBeta}, $\left\vert\beta_a^t \right\vert$ can be upper-bounded as follows:
\begin{align}
    \left\vert\beta_a^t \right\vert \leq \left\Vert c_a \right\Vert_2 \left\Vert (\Gamma - \Gamma K C_\theta)^s \right\Vert_2 \left\Vert \hat{z}_{t-s} \right\Vert_2,
\end{align}
which according to Assumption \ref{assum:action_bound} and Lemma \ref{lemma:subgaussian_half-normal}, it can be shown that the following inequality holds with a probability of at least $1-\delta$:
\begin{equation}\label{eq:MMM1}
    \left\vert\beta_a^t \right\vert \leq B_c \left\Vert(\Gamma - \Gamma K C_\theta)^s \right\Vert_2 \sqrt{g_{\Sigma_{\hat{z}_{t-s}}}\left(\delta\right)}, 
\end{equation} 
where $\Sigma_{\hat{z}_{t-s}} \triangleq \mathbb{E}\left[\hat{z}_{t-s}\hat{z}_{t-s}^\top \right]$.

 In what follows, we will provide a bound on $\mathbb{E}\left[\hat{z}_{t-s}\hat{z}_{t-s}^\top \right]$ and consequently $\mathbb{E}\left[\left\Vert \hat{z}_{t-s} \right\Vert_2^2 \right]$. Let $t > 0$ be any arbitrary value. Since $z_t = \hat{z}_t + e_t$, orthogonality principle \cite{Sayed2008} implies that $\mbox{tr}\left(\mathbb{E}\left[\hat{z}_te_t^\top + e_t\hat{z}_t^\top \right]\right) = 0$. Thus, we have: 
\begin{align}
\mathbb{E}\left[z_t z_t^\top \mid \mathscr{F}_{t-1}\right]  = & \mathbb{E}\left[\left(\hat{z}_t + e_t\right)\left(\hat{z}_t + e_t\right)^\top \mid \mathscr{F}_{t-1}\right] \nonumber \\
= &\mathbb{E}\left[\hat{z}_t\hat{z}_t^\top + \hat{z}_t e_t^\top + e_t\hat{z}_t^\top + e_te_t^\top \mid \mathscr{F}_{t-1}\right] \nonumber\\
= &\mathbb{E}\left[\hat{z}_t\hat{z}_t^\top \mid \mathscr{F}_{t-1}\right]  +  \mathbb{E}\left[e_te_t^\top \mid \mathscr{F}_{t-1}\right]\nonumber\\
&+ \mathbb{E}\left[\hat{z}_t e_t^\top + e_t\hat{z}_t^\top \mid \mathscr{F}_{t-1}\right]  \nonumber\\
= &\mathbb{E}\left[\hat{z}_t\hat{z}_t^\top \mid \mathscr{F}_{t-1}\right] + \mathbb{E}\left[e_te_t^\top \mid \mathscr{F}_{t-1}\right] \nonumber,
\end{align}
which implies that $Z_t = \mathbb{E}\left[z_t z_t^\top \right] = \mathbb{E}\left[\mathbb{E}\left[z_t z_t^\top \mid \mathscr{F}_{t-1}\right]\right] = \hat{Z}_t + P_t$, where $\hat{Z}_t \triangleq \mathbb{E}\left[\hat{z}_t\hat{z}_t^\top \right]$ and $P_t \triangleq \mathbb{E}\left[e_te_t^\top \mid \mathscr{F}_{t-1} \right]$. Since $\hat{Z}_t \succeq 0$ and $P_t \succeq 0$, it can be easily shown that 
\begin{equation}
    \sqrt{\mbox{tr}\left(\hat{Z}_t\right)} \leq \sqrt{\mbox{tr}\left(Z_t\right)} \overset{(a)}{\Rightarrow}\sqrt{\mbox{tr}\left(\hat{Z}_{t-s}\right)}  \leq \sqrt{\mbox{tr}\left(Z_{t-s}\right)}  \nonumber, 
\end{equation}
where in $(a)$ we used the implication that $t > 0$ is any arbitrary value. Thus, it is concluded that the following inequality holds with a probability of at least $1-\delta$:
\begin{equation}\label{eq: beta_bound}
    \left\vert\beta_a^t  \right\vert \leq  B_{\beta}\left(s\right)\sqrt{g_{\Sigma_{z_t}}\left(\delta\right)},
\end{equation}
where $B_{\beta}\left(s\right)$ is defined in \eqref{eq:B_beta_def}.

 Regarding the term $\sqrt{\Theta_{t}^\top V_a^t\left(s\right)^{-1}\Theta_{t}} $ in \eqref{eq:B_T_norm_2}, we have:
\begin{align}
& \left\Vert\mathbf{O}_{\mathscr{T}_a}\mathbf{B}_{\mathscr{T}_a}^\top \right\Vert_{V_a^t\left(s\right)^{-1}} \nonumber \\
& ~~ \leq \sum_{t \in \mathscr{T}_a} B_{\beta}\left(s\right)\sqrt{g_{\Sigma_{z_t}}\left(\delta\right)}\sqrt{\mbox{tr}\left(V_a^t\left(s\right)^{-1}\Theta_{t}\Theta_{t}^\top \right) \cdot 1}  \nonumber \\
% & ~~ \overset{(b)}{\leq} B_{\beta}\left(s\right)\sqrt{2\log\left(1/\delta\right)}\sqrt{\sum_{t \in \mathscr{T}_a}\mbox{tr}\left(V_a^t\left(s\right)^{-1}\Theta_{t}\Theta_{t}^\top \right)}  \cdot \sqrt{\sum_{t \in \mathscr{T}_a} 1} \nonumber \\
& ~~ \overset{(b)}{\leq} \sqrt{N_a} B_{\beta}\left(s\right)\sqrt{g_{\Sigma_{z_t}}\left(\delta\right)}\sqrt{\sum_{t \in \mathscr{T}_a}\mbox{tr}\left(V_a^t\left(s\right)^{-1}\Theta_{t}\Theta_{t}^\top \right)} \nonumber,% \\
% & ~~ \leq \sqrt{N_a} B_{\beta}\left(s\right)\sqrt{2\log\left(1/\delta\right)}\sqrt{\mbox{tr}\left(V_a^t\left(s\right)^{-1} \sum_{t \in \mathscr{T}_a} \Theta_{t}\Theta_{t}^\top \right)} ,
\end{align}
\begin{multline}\label{eq:InequalityMehdi1} 
    \Rightarrow \left\Vert\mathbf{O}_{\mathscr{T}_a}\mathbf{B}_{\mathscr{T}_a}^\top \right\Vert_{V_a^t\left(s\right)^{-1}} \leq \\
    \sqrt{N_a} B_{\beta}\left(s\right)\sqrt{g_{\Sigma_{z_t}}\left(\delta\right)}\sqrt{\mbox{tr}\left(V_a^t\left(s\right)^{-1} \sum_{t \in \mathscr{T}_a} \Theta_{t}\Theta_{t}^\top \right)}, 
\end{multline}
where $(b)$ is concluded based on the Cauchy-Schwarz inequality. Since $\sum_{\tau \in \mathscr{T}_a} \Theta_{\tau} \Theta_{\tau}^\top = V_a^t\left(s\right) - \lambda I_{ms+1}$, we can rewrite \eqref{eq:InequalityMehdi1} as follows which is satisfied with a probability of at least $1-N_a\delta$:
\begin{multline}
    \left\Vert\mathbf{O}_{\mathscr{T}_a}\mathbf{B}_{\mathscr{T}_a}^\top \right\Vert_{V_a^t\left(s\right)^{-1}} \leq\sqrt{N_a g_{\Sigma_{z_t}}\left(\delta\right)} B_{\beta}\left(s\right)\cdot \\
    \sqrt{\mbox{tr}\left(V_a^t\left(s\right)^{-1} \left(V_a^t\left(s\right) - \lambda I_{ms+1}\right)\right)} \label{eq:matrix_V_bound}. 
\end{multline} 

Note that the matrix $\sum_{\tau \in \mathscr{T}_a} \Theta_{\tau} \Theta_{\tau}^\top \succeq 0$, implying that  $V_a^t\left(s\right) \succeq \lambda I_{ms+1}$ . Thus, the terms in the square root of inequality \eqref{eq:matrix_V_bound} is always non-negative. According to \eqref{eq: beta_bound} and \eqref{eq:matrix_V_bound} and setting $\delta = \delta/N_a$ no further effort is needed to show that inequality \eqref{eq:bias_bound_matrix} is satisfied with a probability of at least $1-\delta$.
\end{proof}

Using Lemma \ref{lemma:bias_error}, the following theorem completes \textbf{STEP 1}, which is written as Theorem \ref{theorem:model_error}. The proof for Theorem \ref{theorem:model_error} is below.

\begin{proof}
From \eqref{eq:linear_model_time_varying} and \eqref{eq:identify_2}, we have:
\begin{align}
\hat{G}_a^t\left( s\right)^\top =& G_{a}\left(s\right)^\top\mathbf{O}_{\mathscr{T}_a} \mathbf{O}_{\mathscr{T}_a}^\top V_a^t\left( s\right)^{-1} \nonumber\\
&+ \left(\mathbf{B}_{\mathscr{T}_a} + \mathbf{E}_{\mathscr{T}_a}\right)\mathbf{O}_{\mathscr{T}_a}^\top V_a^t\left( s\right)^{-1}.\label{eq:equalityproof1} 
\end{align}

Adding and subtracting the term  $\lambda G_{a}\left(s\right)^\top V_a^t\left(s\right)^{-1}$ from the right-hand side of \eqref{eq:equalityproof1} yields:
\begin{align}
\hat{G}_a^t\left( s\right)^\top =& G_{a}\left(s\right)^\top\left(\mathbf{O}_{\mathscr{T}_a} \mathbf{O}_{\mathscr{T}_a}^\top + \lambda I\right)V_a^t\left( s\right)^{-1} \nonumber\\
&+ \left(\mathbf{B}_{\mathscr{T}_a} + \mathbf{E}_{\mathscr{T}_a}\right)\mathbf{O}_{\mathscr{T}_a}^\top V_a^t\left( s\right)^{-1} \nonumber\\
&- \lambda G_{a}\left(s\right)^\top V_a^t\left( s\right)^{-1}, \nonumber 
\end{align}
which according to the fact that  $V_a^t\left(s\right) =\lambda I + \mathbf{O}_{\mathscr{T}_a} \mathbf{O}_{\mathscr{T}_a}^\top$, implies that:
\begin{multline}
\hat{G}_a^t\left( s\right)^\top  - G_{a}\left(s\right)^\top = - \lambda G_{a}\left(s\right)^\top V_a^t\left( s\right)^{-1} \\ + \left(\mathbf{B}_{\mathscr{T}_a} + \mathbf{E}_{\mathscr{T}_a} \right) \mathbf{O}_{\mathscr{T}_a}^\top V_a^t\left( s\right)^{-1}.\label{eq:equality1}
\end{multline}

Multiplying both sides of this last equation by $V_a^t\left(s\right)$ from right yields:
\begin{multline}
\Big(\hat{G}_a^t\left(s\right)  - G_{a}\left(s\right) \Big)^\top V_a^t\left(s\right) \\ 
= \left(- \lambda G_{a}\left(s\right) + 
\mathbf{O}_{\mathscr{T}_a}\left(\mathbf{B}_{\mathscr{T}_a} + \mathbf{E}_{\mathscr{T}_a} \right)^\top \right)^\top. \label{eq:equality2}
\end{multline}

Multiplying the right- and left-hand sides of \eqref{eq:equality2} by transpose of the right- and left-hand sides of \eqref{eq:equality2}, respectively, and according to the fact that  $V_a^t\left(s\right)^{-1}=\left(V_a^t\left(s\right)^{-1}\right)^\top$, we obtain the following:
\begin{multline}\label{eq:Mehdi1}
\Big\Vert\hat{G}_a^t \left( s\right) - G_{a}\left(s\right)\Big\Vert_{V_a^t\left(s\right)}  \\ 
=\left\Vert- \lambda G_{a}\left(s\right) + \mathbf{O}_{\mathscr{T}_a}\left(\mathbf{B}_{\mathscr{T}_a} + \mathbf{E}_{\mathscr{T}_a} \right)^\top \right\Vert_{V_a^t\left(s\right)^{-1}} .
\end{multline}

which implies that\footnote{Given $z_1,z_2\in\mathbb{R}^n$ and $Q\succ0$ ($Q\in\mathbb{R}^{n\times n}$), we have $\left\Vert z_1+z_2\right\Vert_Q^2=z_1^\top Q^{\frac{1}{2}}Q^{\frac{1}{2}}z_1+z_2^\top Q^{\frac{1}{2}}Q^{\frac{1}{2}}z_2+2z_1^\top Q^{\frac{1}{2}}Q^{\frac{1}{2}}z_2\leq\left\Vert Q^{\frac{1}{2}}z_1\right\Vert^2+\left\Vert Q^{\frac{1}{2}}z_2\right\Vert^2+2\left\Vert Q^{\frac{1}{2}}z_1\right\Vert\left\Vert Q^{\frac{1}{2}}z_2\right\Vert$, which implies that $\left\Vert z_1+z_2\right\Vert_Q^2\leq\left(\left\Vert Q^{\frac{1}{2}}z_1\right\Vert+\left\Vert Q^{\frac{1}{2}}z_2\right\Vert\right)^2$, or $\left\Vert z_1+z_2\right\Vert_Q^2\leq\left(\left\Vert z_1\right\Vert_Q+\left\Vert z_2\right\Vert_Q\right)^2$. Thus, we have $\left\Vert z_1+z_2\right\Vert_Q\leq\left\Vert z_1\right\Vert_Q+\left\Vert z_2\right\Vert_Q$.}
\begin{multline}\label{eq:model_error}        
\left\Vert\hat{G}_a^t \left( s\right) - G_{a}\left(s\right)\right\Vert_{V_a^t\left(s\right)}  \leq \left\Vert\lambda G_{a}\left(s\right) \right\Vert_{V_a^t\left(s\right)^{-1}}  \\ + \left\Vert\mathbf{O}_{\mathscr{T}_a}\mathbf{B}_{\mathscr{T}_a}^\top \right\Vert_{V_a^t\left(s\right)^{-1}} + \left\Vert \mathbf{O}_{\mathscr{T}_a}\mathbf{E}_{\mathscr{T}_a} ^\top \right\Vert_{V_a^t\left(s\right)^{-1}} . 
\end{multline}

Since $\varepsilon_{a}^t$ is conditionally $B_R$-subgaussian on $\mathscr{F}_{t-1}$ and $\Theta_{t}$ is $\mathscr{F}_{t-1}$ measurable, it can be concluded that \cite[Theorem 1]{NIPS2011_e1d5be1c} there exists $\delta \in (0,1)$ such that the following inequality is satisfied with a probability of at least $1-\delta$:
\begin{multline}\label{eq:boundproof1}
\left\Vert \mathbf{O}_{\mathscr{T}_a}\mathbf{E}_{\mathscr{T}_a}^\top \right\Vert_{V_a^t\left(s\right)^{-1}} \leq \\ \sqrt{2B_R^2\log\left(\frac{1}{\delta}\frac{ \det(V_a^t\left(s\right))^{1/2}}{\det(\lambda I)^{1/2}}\right)} .
\end{multline}

Finally, according to Assumption \ref{assum:G_a_Assumption}, we have:
\begin{align}\label{eq:InequalityMehdi2}
\left\Vert\lambda G_{a}\left(s\right) \right\Vert_{V_a^t\left(s\right)^{-1}} & \leq \lambda \left\Vert V_a^t\left( s\right)^{-1/2} \right\Vert_2\left\Vert G_{a}\left(s\right) \right\Vert_2 \nonumber \\
& \leq \lambda B_G  \sqrt{\mbox{tr}\left(V_a^t\left(s\right)^{-1}\right)} . 
\end{align}

Thus, from Lemma \ref{lemma:bias_error}, and according to inequalities \eqref{eq:boundproof1} and \eqref{eq:InequalityMehdi2}, one can obtain the bound given in \eqref{eq:probabilistic_bound}, which completes the proof.  
\end{proof}

\subsection{STEP 2}\label{appendix:step_2}

Using results from \textbf{STEP 1}, Lemma \ref{lemma:optimal_value} will complete \textbf{STEP 2}. The proof for Lemma \ref{lemma:optimal_value} is below.

\begin{proof}
Consider the following optimization problem 
\begin{equation}\label{eq:optimism_optim}
    \begin{array}{cc}
        \underset{\tilde{G} \in \mathbb{R}^{(ms+1) \times 1}}{\max} & \tilde{G}^\top  \Theta_{t} \\
            \mbox{ s.t. } & \tilde{G} \in \mathscr{C}_{a,s}
        \end{array},
\end{equation}
where the set $\mathscr{C}_{a,s}$ is defined as follows: 
\begin{equation}\label{eq:DefinitionMehdi1}
\mathscr{C}_{a,s} = \left\{\tilde{G}\big|\left\Vert \hat{G}_a^t \left( s\right) - \tilde{G} \right\Vert_{V_a^t\left(s\right)}^2  \leq b_a^t(\delta,s)^2\right\},
\end{equation}
with $b_a^t\left(\delta,s\right)$ as in \eqref{eq:probabilistic_bound}. The Lagrangian $L(\tilde{G},\epsilon)$ for optimization problem \eqref{eq:optimism_optim} is 
\begin{multline}
L\left(\tilde{G},\epsilon\right) = \tilde{G}^\top \Theta_{t} \\ + \epsilon \left(\left\Vert\hat{G}_a^t \left( s\right) - \tilde{G}\right\Vert_{V_a^t\left(s\right)}^2- b_a^t \left(\delta,s\right)^2\right),
\end{multline}
where $\epsilon$ is Lagrange multiplier.

% {\color{red}Comment (Mehdi):\\
% 1) The original constraint is as $g(x)\leq b$, for some $b$. If you write it like $g^2(x)\leq b^2$, these two are not equivalent. Because, the second one means that $-b\leq g(x)\leq b$, which is different than the original constraint.\\
% 2) When you have optimization problem in the form $\min f(x)$ subject to $g(x)\leq b$, the correct Lagrangian form is $f(x)+\lambda(g(x)-b)$, not $f(x)-\lambda(g(x)-b)$. The way you wrote, implies that the Lagrangian multiplier can be negative, which is not correct.}

The partial derivatives of $L\left(\tilde{G},\epsilon\right)$ are 
\begin{align}
\frac{\partial L(\cdot)}{\partial \tilde{G}^\top} = \Theta_{t}^\top + 2\epsilon \tilde{G}^\top V_a^t \left( s\right)-2\epsilon \hat{G}_a^t\left( s\right)^\top V_a^t \left( s\right)\nonumber,
\end{align}
and
\begin{align}
\frac{\partial L(\cdot)}{\partial \epsilon}  = \left\Vert\hat{G}_a^t \left( s\right) - \tilde{G}\right\Vert_{V_a^t\left(s\right)} ^2-  b_a^t \left(\delta,s\right)^2.\label{eq:partial_mu}
\end{align}

Solving $\partial L\left(\tilde{G},\epsilon\right)/\partial \tilde{G}^\top = 0$ for $\tilde{G}$ yields 
\begin{align}
\Theta_{t}^\top + 2\epsilon \tilde{G}^\top V_a^t \left( s\right) - 2\epsilon \hat{G}_a^t\left( s\right)^\top V_a^t \left( s\right) = 0,
\end{align}
which implies that
\begin{align}
\tilde{G}^\top = \hat{G}_a^t\left( s\right)^\top -  \frac{\Theta_t^\top V_a^t\left(s\right)^{-1}}{2\epsilon} . \label{eq:mu_solved}
\end{align}

Using $\tilde{G}$ given in \eqref{eq:mu_solved} in \eqref{eq:partial_mu}, and solving $\partial L\left(\tilde{G},\epsilon\right)/\partial \epsilon = 0$ for $\epsilon$, and noting that $\epsilon$ is non-negative, yields the optimal value of $\epsilon$:
\begin{equation}
\epsilon = \frac{\sqrt{\Theta_t^\top V_a^t\left(s\right)^{-1}\Theta_t}}{2b_a^t\left(\delta,s\right)} .
\end{equation}

Plugging the obtained $\epsilon$ in \eqref{eq:mu_solved} yields
\begin{align}
    \tilde{G}^\top  & = \hat{G}_a^t\left( s\right)^\top  -  \frac{\Theta_t^\top V_a^t\left(s\right)^{-1}}{2\epsilon} \nonumber \\
    & = \hat{G}_a^t\left( s\right)^\top   + \frac{\Theta_t^\top V_a^t\left(s\right)^{-1}}{2} \left( \frac{\sqrt{\Theta_t^\top V_a^t\left(s\right)^{-1}\Theta_t}}{2b_a^t\left(\delta,s\right)} \right)^{-1},  \nonumber
\end{align}
where this solution lies on the edge of the constrained set $\mathscr{C}_{a,s}$ in optimization problem \eqref{eq:optimism_optim}. Since this optimal value lies on the edge of the confidence interval $\mathscr{C}_{a,s}$ given in \eqref{eq:DefinitionMehdi1}, then with a probability of at least $1-2\delta$, we have 
\begin{multline}
    \hat{G}_a^t\left( s\right)^\top \Theta_t - b_a^t\left(\delta,s\right) \sqrt{\Theta_{t}^\top V_a^t\left(s\right)^{-1} \Theta_{t}}  \leq G_{a}\left(s\right)^\top \Theta_t \\ \leq \hat{G}_a^t\left( s\right)^\top \Theta_t + b_a^t\left(\delta,s\right) \sqrt{\Theta_{t}^\top V_a^t\left(s\right)^{-1} \Theta_{t}} \nonumber,
\end{multline}
which yields the upper-bound \eqref{eq:InequalityJon} and completes the proof.

% {\color{red}Comment (Mehdi):\\
% 1) Since $\epsilon$ is the Lagrangian multiplier, it cannot be negative. Thus, we cannot have $\pm$. I changed everything to +. Make sure it does not impact the remaining of the proof.\\
% 2) ``Since this optimal value lies on the edge of the confidence interval $\mathscr{C}_{a,s}$ given in \eqref{eq:DefinitionMehdi1}". Why? This can be correct if we have equality constraint. But for inequality constraint this is not necessarily correct.}
\end{proof}

\section{Proofs for Theorem \ref{theorem:regret} and Corollary \ref{corollary:regret_bound}}\label{appendix:regret_bounds}

To provide more details on the regret bound derived in Theorem \ref{theorem:regret} and Corollary \ref{corollary:regret_bound}, we provide a more comprehensive theorem and corollary with their proofs written below. 

\begin{theorem}\label{theorem:regret_2}
    Assume that $n > 0$ is finite. Let actions $a \in [k]$ be selected based on \eqref{eq:perturbation_method} where is defined as \eqref{eq:u_a_b}, $\mathbf{b}_a^t\left(s\right)$ is defined as \eqref{eq:b_def}, and $\mathbf{B}_{\beta} \geq B_{\beta}\left(s\right)$. Regret $\sum_{t=1}^n X_t^* - X_t$ has the following  upper-bound with a probability of at least $1- 9\delta$: 
        \begin{multline}\label{eq:regret_bound_OFU}
            R_n^{(\mathbf{b})} \leq
            k\sqrt{n}\overline{B}_{\theta}^\delta\left(n\right)\overline{B}_u^\delta\left(n\right)\left(\overline{B}_{(\mathbf{b})}^\delta\left(n\right)+ \overline{\mathbf{B}}_{(\mathbf{b})}^\delta\left(n\right)\right) \\ + 2knB_{\beta}\left(s\right)\sqrt{2\log\left(n/\delta\right)}  + k\sqrt{4n B_R^2 \log\left(1/\delta\right)}, 
        \end{multline}
        
        The variables $\overline{B}_{\theta}^\delta\left(n\right)$, $\overline{B}_u^\delta\left(n\right)$, $\overline{B}_{(\mathbf{b})}^\delta\left(n\right)$, $\overline{\mathbf{B}}_{(\mathbf{b})}^\delta\left(n\right)$, and $L_{\Theta}\left(\delta\right)$ are defined as
        \begin{multline}\label{eq:b_theta_bound_def}
            \overline{B}_{\theta}^\delta\left(n\right)\triangleq \\ \sqrt{\frac{2\left(\mbox{tr}\left(\mathbf{C}_s \hat{Z}_n \mathbf{C}_s^\top + \Sigma_E \right)+1\right)\log\left(n/\delta\right)}{\lambda}}, 
        \end{multline}
        \begin{multline}\label{eq:b_u_bound}
            \overline{B}_u^\delta\left(n\right) \triangleq \\ \sqrt{2\left(ms+1\right)\log\left(\frac{\lambda\left(ms+1\right) +  nL_{\Theta}\left(\delta\right)}{\lambda\left(ms+1\right)}\right)}, 
        \end{multline}
        \begin{multline}\label{eq:large_b_bound}
            \overline{B}_{(\mathbf{b})}^\delta\left(n\right) \triangleq \\ \sqrt{2B_R^2\log\left(\frac{n}{\delta}\left(\frac{\lambda \left(ms+1\right) + n L_{\Theta}\left(\delta\right)}{\lambda \left(ms+1\right)}\right)^{\frac{ms+1}{2}}\right)} \\+ \sqrt{n\left(ms+1\right)g_{\Sigma_{z_t}}\left(\delta/n\right)} B_{\beta}\left(s\right)\\+ \lambda B_G \sqrt{\frac{ms+1}{\lambda}} , 
        \end{multline}
        \begin{multline}\label{eq:large_b_bound_bold}
                \overline{\mathbf{B}}_{(\mathbf{b})}^\delta\left(n\right) \triangleq \\ \sqrt{2B_R^2\log\left(\frac{n}{\delta}\left(\frac{\lambda \left(ms+1\right) + n L_{\Theta}\left(\delta\right)}{\lambda \left(ms+1\right)}\right)^{\frac{ms+1}{2}}\right)} \\+ \sqrt{n\left(ms+1\right)g_{\Sigma_{z_t}}\left(\delta/n\right)} \mathbf{B}_{\beta} \\ + \lambda B_G \sqrt{\frac{ms+1}{\lambda}} , 
        \end{multline}
        \begin{equation}
            L_{\Theta}\left(\delta\right) \triangleq n\frac{\mbox{tr}\left(\Sigma_E + \mathbf{C}_s\hat{Z}_n \mathbf{C}_s\right)+1}{\delta}  \label{eq:L_bound}, 
        \end{equation}
    and terms $\Sigma_E$ and $\mathbf{C}_s \hat{Z}_t\mathbf{C}_s^\top$ are defined as follows:    
    \begin{align}
        \Sigma_E & \triangleq \mathbf{F}_s\begin{bmatrix}
            \Sigma_\theta & \dots & \mathbf{0} \\
            \vdots & \ddots & \vdots \\
            \mathbf{0} & \dots & \Sigma_\theta
        \end{bmatrix}\mathbf{F}_s^\top \in \mathbb{R}^{ms \times ms} \label{eq:sigma_e_def} \\
        \Sigma_\theta & \triangleq C_\theta P C_\theta^\top + R \label{eq:theta_residual} \\
        \hat{Z}_t & \triangleq \mathbb{E}\left[\hat{z}_t \hat{z}_t^\top \right]\label{eq:sigma_zhat} \\
        \hat{Z}_{t+1} & = \Gamma \hat{Z}_t \Gamma^\top + \Gamma K \Sigma_\theta K^\top \Gamma^\top  \nonumber, 
    \end{align} 
    \begin{align}
        \mathbf{F}_s & \triangleq \begin{bmatrix}
            I_m & \mathbf{0}_{m \times m} & \dots & \mathbf{0}_{m \times m} & \mathbf{0}_{m \times m} \\
            C_\theta \Gamma K & I_m & \dots & \mathbf{0}_{m \times m} & \mathbf{0}_{m \times m} \\
            \vdots & \vdots & \ddots & \vdots & \vdots \\
            C_\theta\Gamma^{s-3}K & C_\theta\Gamma^{s-4}K & \dots & I_m & \mathbf{0}_{m \times m} \\
            C_\theta\Gamma^{s-2}K & C_\theta\Gamma^{s-3}K & \dots & C_\theta\Gamma K & I_m
        \end{bmatrix} \nonumber \\
        \mathbf{C}_s & \triangleq \begin{bmatrix}
            C_\theta^\top &
            \Gamma^\top C_\theta^\top  &
            \dots &
            \left(\Gamma^{s-2}\right)^\top C_\theta^\top  &
            \left(\Gamma^{s-1}\right)^\top C_\theta^\top 
        \end{bmatrix}^\top \nonumber. 
    \end{align}

    Therefore, the upper bound for regret increases at the following rate \eqref{eq:regret_rate_increase} with a probability of at least $1 - 9\delta$.  
\end{theorem} 

\begin{proof}
According to \eqref{eq:Matrix_Form_2}, instantaneous regret $r_t^{(\mathbf{b})} \triangleq X_t^* - X_t$ has the following expression
\begin{align}
        r_t^{(\mathbf{b})} =& X_t^\ast-X_t\nonumber\\
        =&G_{a^*}\left(s\right)^\top \Theta_t - G_a\left(s\right)^\top \Theta_t + \beta_{a^*}^t - \beta_a^t \nonumber \\
        & ~~~~~~~~~~~~~~~~~~ + \varepsilon_{a^*}^t - \varepsilon_a^t. \label{eq:inequality_regret_1}
\end{align}

First, let us prove the regret bound for when the learner chooses actions $a \in [k]$ based on \eqref{eq:perturbation_method} where $u_a \equiv u_a^{(\mathbf{b})}$ \eqref{eq:u_a_b}: On the one hand, based on Lemma \ref{lemma:optimal_value}, with a probability of at least $1-2\delta$, the  following inequality is satisfied:
\begin{multline}
G_{a^*}\left(s\right)^\top \Theta_t \leq \hat{G}_{a^*}^t\left(s\right)^\top \Theta_t \\ + \mathbf{b}_{a^*}^t\left(\delta,s\right)\sqrt{\Theta_{t}^\top V_{a^*}^t\left(s\right)^{-1} \Theta_{t}}, \nonumber
\end{multline}
\begin{equation}\label{eq:upper_bound_ofu}
    \Rightarrow G_{a^*}\left(s\right)^\top \Theta_t \leq \hat{G}_{a^*}^t\left(s\right)^\top \Theta_t + u_{a^*}^{(\mathbf{b})}. 
\end{equation}

On the other hand, action $a$ is chosen because of the following relationship 
\begin{equation}
\hat{G}_{a^*}^t\left(s\right)^\top \Theta_t + u_{a^*}^{(\mathbf{b})} \leq 
\hat{G}_{a}^t\left(s\right)^\top \Theta_t + u_a^{(\mathbf{b})}. \nonumber
\end{equation}

Therefore, inequality \eqref{eq:inequality_regret_1} can be upper-bounded by 
    \begin{multline}\label{eq:inequality_regret_2}
        r_t^{(\mathbf{b})} \leq \hat{G}_{a}^t\left(s\right)^\top \Theta_t + u_a^{(\mathbf{b})} \\ - G_a\left(s\right)^\top \Theta_t + \beta_{a^*}^t - \beta_a^t + \varepsilon_{a^*}^t - \varepsilon_a^t .
    \end{multline}

    By Lemma \ref{lemma:bias_error}, the difference $\hat{G}_{a}^t\left(s\right)^\top \Theta_t -  G_a\left(s\right)^\top \Theta_t$ has the following upper bound with a probability of at least  $1-2\delta$
    \begin{equation}
        \hat{G}_{a}^t\left(s\right)^\top \Theta_t -  G_a\left(s\right)^\top \Theta_t \leq b_a^t\left(\delta,s\right)\sqrt{\Theta_{t}^\top V_a^t\left(s\right)^{-1} \Theta_{t}}\nonumber,
    \end{equation}
    adjusting bound \eqref{eq:inequality_regret_2} to be 
    \begin{multline}\label{eq:inequality_single_terms}
        r_t^{(\mathbf{b})} \leq b_a^t\left(\delta,s\right)\sqrt{\Theta_{t}^\top V_a^t\left(s\right)^{-1} \Theta_{t}} + u_a^{(\mathbf{b})} \\ + \beta_{a^*}^t - \beta_a^t + \varepsilon_{a^*}^t - \varepsilon_a^t  .
    \end{multline}

To upper bound the term $b_a^t\left(\delta,s\right)\sqrt{\Theta_{t}^\top V_a^t\left(s\right)^{-1} \Theta_{t}}$, first we show that $b_a^t\left(\delta,s\right)\leq \overline{B}_{(\mathbf{b})}^\delta\left(N_a\right)$, where $\overline{B}_{(\mathbf{b})}^\delta\left(N_a\right)$ is as in \eqref{eq:large_b_bound}:
\begin{itemize}
\item First term in $\overline{B}_{(\mathbf{b})}^\delta\left(N_a\right)$: According to the Markov inequality \cite{lattimore2020bandit}, the random variable $\left\Vert \Theta_t \right\Vert_2^2 \leq \mathbb{E}\left[\left\Vert \Theta_t \right\Vert_2^2\right]/\delta \Rightarrow \left\Vert \Theta_t \right\Vert_2^2 \leq L_{\Theta}\left(\delta\right)/n$ with a probability of at least $1-\delta$. Therefore, for all rounds $t = 1,\dots,n$, since $\left\Vert \Theta_{t}\right\Vert_2^2 \leq L_{\Theta}\left(\delta\right)$, then with a probability of at least $1-\delta$, Lemma 19.4 of \cite{lattimore2020bandit} implies that:
\begin{multline}\label{eq:UpperBoundb1}
    2B_R^2\log\left(\frac{n}{\delta}\frac{ \det(V_a^t\left(s\right))^{1/2}}{\det(\lambda I)^{1/2}}\right) \leq 2B_R^2\log\left(n/\delta\right) \\ + 2B_R^2\frac{ms+1}{2}\log\left(\frac{\mbox{tr}\left(\lambda I_{ms+1}\right) + N_a L_{\Theta}\left(\delta\right)}{ms+1}\right) \\ - 2B_R^2\left(\frac{1}{2}\log\det\left(\lambda I_{ms+1}\right)\right);
\end{multline} 
\item Second term in $\overline{B}_{(\mathbf{b})}^\delta\left(N_a\right)$: Since $V_a^t\left(s\right) \succ 0$ implying that $- \lambda V_a^t\left(s\right)^{-1}\prec0$, the following upper-bound can be obtained 
\begin{multline}\label{eq:UpperBoundb2}
    \sqrt{N_a g_{\Sigma_{z_t}}\left(\delta/N_a\right)} B_{\beta}\left(s\right)\sqrt{\mbox{tr}\left(I_{ms+1}- \lambda V_a^t\left(s\right)^{-1}\right)}  \\
    \leq\sqrt{N_a(ms+1)g_{\Sigma_{z_t}}\left(\delta/N_a\right)} B_{\beta}\left(s\right);
\end{multline}
% \begin{align}\label{eq:UpperBoundb2}
% &\sqrt{N_a \log\left(1/\delta\right)} B_{\beta}\sqrt{\mbox{tr}\left(I_{ms+1}- \lambda V_a^t\left(s\right)^{-1}\right)} \leq\nonumber\\ 
% &~~~~~~~~~~~~~~~~~~~~~~~\sqrt{N_a\sqrt{ms+1} \log\left(1/\delta\right)}  B_{\beta};
% \end{align}

\item Third term in $\overline{B}_{(\mathbf{b})}^\delta\left(N_a\right)$: The matrix $V_a^t \left(s\right) \succeq \lambda I_{ms+1}$, implying that $V_a^t \left(s\right)^{-1} \preceq I_{ms+1}/\lambda$. Therefore, it can be shown that the following inequality is satisfied:
\begin{equation}
\lambda B_G \sqrt{\mbox{tr}\left(V_a^t\left(s\right)^{-1}\right)} \leq \lambda B_G \sqrt{\frac{ms+1}{\lambda}} \nonumber. 
\end{equation}
\end{itemize}

Combining all the terms provides the bound $\overline{B}_{(\mathbf{b})}^\delta\left(N_a\right)$ in \eqref{eq:large_b_bound}. The logic follows for $\mathbf{b}_a^t\left(\delta,s\right)$ \eqref{eq:b_def} in $u_a^{(\mathbf{b})}$ where instead of using $B_{\beta}\left(s\right)$ we use $\mathbf{B}_{\beta}$. Next, using the sum of terms in \eqref{eq:inequality_single_terms} with $\overline{B}_{(\mathbf{b})}^\delta\left(N_a\right)$ \eqref{eq:large_b_bound} and $\overline{\mathbf{B}}_{(\mathbf{b})}^\delta\left(N_a\right)$ \eqref{eq:large_b_bound_bold} in place of $b_a^t\left(\delta,s\right)$ \eqref{eq:probabilistic_bound} and $\mathbf{b}_a^t\left(\delta,s\right)$ \eqref{eq:b_def}, we have
\begin{multline}
    \sum_{t=1}^n r_t^{(\mathbf{b})} \leq \sum_{a \in [k]}\sum_{t \in \mathscr{T}_a} \beta_{a^*}^t - \beta_a^t + \varepsilon_{a^*}^t - \varepsilon_a^t + \\ \sum_{a \in [k]} \left(\overline{B}_{(\mathbf{b})}^\delta\left(N_a\right) + \overline{\mathbf{B}}_{(\mathbf{b})}^\delta\left(N_a\right) \right)\sum_{t \in \mathscr{T}_a} \sqrt{\Theta_t^\top V_a^t\left(s\right)^{-1}\Theta_t} \nonumber. 
\end{multline}

For the summation $\sum_{t \in \mathscr{T}_a} \sqrt{\Theta_t^\top V_a^t\left(s\right)^{-1}\Theta_t}$, we can bound this as follows. First, since $V_a^t\left(s\right) = \lambda I_{ms+1} + \sum_{t \in \mathscr{T}_a} \Theta_t \Theta_t^\top$, then we know that 
\begin{equation}
    \sqrt{\Theta_t^\top V_a^t\left(s\right)^{-1}\Theta_t} \leq \min\left\{\frac{\left\Vert \Theta_t \right\Vert_2}{\sqrt{\lambda}},\sqrt{\Theta_t^\top V_a^t\left(s\right)^{-1}\Theta_t}\right\} \nonumber, 
\end{equation}
\begin{multline}
    \Rightarrow \sum_{t \in \mathscr{T}_a} \sqrt{\Theta_t^\top V_a^t\left(s\right)^{-1}\Theta_t} \\
    \leq \sum_{t \in \mathscr{T}_a} \min\left\{\frac{\left\Vert \Theta_t \right\Vert_2}{\sqrt{\lambda}},\sqrt{\Theta_t^\top V_a^t\left(s\right)^{-1}\Theta_t}\right\} \nonumber.
\end{multline}

Next, according to Lemma \ref{lemma:subgaussian_half-normal}, the random variable $\left\Vert \Theta_t \right\Vert_2$ is bounded by a subgaussian random variable. Therefore, with a probability of at least $1-\delta$, we satisfy the following inequality 
\begin{equation}
    \left\Vert \Theta_t \right\Vert_2
    \leq \sqrt{g_{\Sigma_{\Theta_t}}\left(\delta/n\right)} \nonumber, 
\end{equation}
\begin{multline}
    \Rightarrow \sum_{t \in \mathscr{T}_a} \sqrt{\Theta_t^\top V_a^t\left(s\right)^{-1}\Theta_t} \\
    \leq \sum_{t \in \mathscr{T}_a} \min\left\{\overline{B}_{\theta}^\delta\left(n\right),\sqrt{\Theta_t^\top V_a^t\left(s\right)^{-1}\Theta_t}\right\} \nonumber, 
\end{multline}
where $\overline{B}_{\theta}^\delta\left(n\right)$ is defined to be \eqref{eq:b_theta_bound_def}. Factoring out $\overline{B}_{\theta}^\delta\left(n\right)$ in the summation and using Cauchy-Schwarz inequality provides the following bound for $\sum_{t \in \mathscr{T}_a} \sqrt{\Theta_t^\top V_a^t\left(s\right)^{-1}\Theta_t}$: 
\begin{multline}\label{eq:bound_assumption_for_elliptic}
    \sum_{t \in \mathscr{T}_a} \sqrt{\Theta_t^\top V_a^t\left(s\right)^{-1}\Theta_t} \\
    \leq \overline{B}_{\theta}^\delta\left(n\right)\sqrt{N_a} \sqrt{\sum_{t \in \mathscr{T}_a} \min\left\{1,\Theta_t^\top V_a^t\left(s\right)^{-1}\Theta_t\right\}}. 
\end{multline}

Based on the bound of the sum $\sum_{t \in \mathscr{T}_a} \sqrt{\Theta_t^\top V_a^t\left(s\right)^{-1}\Theta_t}$ shown in \eqref{eq:bound_assumption_for_elliptic}, the conditions for Lemma 19.4 in \cite{lattimore2020bandit} apply. According to Lemma 19.4, the elliptic potential lemma, if $\left\Vert \Theta_t \right\Vert_2^2 \leq L_{\Theta}\left(\delta\right)$ for any time $t = 1,\dots,n$, and $\lambda > 0$, then the following inequality applies
\begin{multline}
    \sum_{t \in \mathscr{T}_a} \min\left\{1,\Theta_t^\top V_a^t\left(s\right)^{-1}\Theta_t\right\}\leq \\ 2\left(ms+1\right)\log\left(\frac{\lambda\left(ms+1\right) + N_a L_{\Theta}\left(\delta\right)}{\lambda\left(ms+1\right)}\right)
\end{multline}
\begin{multline}
    \Rightarrow \overline{B}_{\theta}^\delta\left(n\right)\sqrt{N_a} \sqrt{\sum_{t \in \mathscr{T}_a} \min\left\{1,\Theta_t^\top V_a^t\left(s\right)^{-1}\Theta_t\right\}} \leq \\ \sqrt{N_a}\overline{B}_{\theta}^\delta\left(n\right)\overline{B}_u^\delta\left(N_a\right) \nonumber, 
\end{multline}
where $\overline{B}_u^\delta\left(N_a\right)$ is defined to be \eqref{eq:b_u_bound}. Based on the equations above, the bounds for the sum
\begin{equation*}
    \sum_{t \in \mathscr{T}_a} \sqrt{\Theta_t^\top V_a^t\left(s\right)^{-1}\Theta_t},
\end{equation*}
and instantaneous regret are 
\begin{equation}
    \sum_{t \in \mathscr{T}_a} \sqrt{\Theta_t^\top V_a^t\left(s\right)^{-1}\Theta_t} \leq \sqrt{N_a}\overline{B}_{\theta}^\delta\left(n\right)\overline{B}_u^\delta\left(N_a\right) \nonumber, 
\end{equation}
\begin{multline}
    \Rightarrow \sum_{t=1}^n r_t^{(\mathbf{b})} \leq \sum_{a \in [k]}\sum_{t \in \mathscr{T}_a} \beta_{a^*}^t - \beta_a^t + \varepsilon_{a^*}^t - \varepsilon_a^t + \\ \sum_{a \in [k]} \sqrt{N_a}\overline{B}_{\theta}^\delta\left(n\right)\overline{B}_u^\delta\left(N_a\right)\left(\overline{B}_{(\mathbf{b})}^\delta\left(N_a\right) + \overline{\mathbf{B}}_{(\mathbf{b})}^\delta\left(N_a\right) \right)\nonumber. 
\end{multline}

For what regards the term $\sum_{t=1}^n \beta_{a^*}^t - \beta_a^t$ in \eqref{eq:inequality_regret_1} and using Lemma \ref{lemma:bias_error}, it follows from \eqref{eq: beta_bound} that the following inequality holds with a probability of at least $1-2\delta$:
\begin{equation}
\sum_{t\in \mathscr{T}_a} \beta_{a^*}^t - \beta_a^t \leq 2N_aB_{\beta}\left(s\right)\sqrt{g_{\Sigma_{z_t}}\left(\delta/N_a\right)}.
\end{equation}

Finally, the residual sum $\sum_{t=1}^n \varepsilon_{a^*}^t - \varepsilon_a^t$ is $\sqrt{N_a} B_R$-subgaussian such that with a probability of at least $1-\delta$, the following inequality holds:
\begin{equation}\label{eq:UpperboundEpsilon}
\sum_{t\in \mathscr{T}_a} \varepsilon_{a^*}^t - \varepsilon_a^t \leq \sqrt{4N_a B_R^2 \log\left(1/\delta\right)}. 
\end{equation}

According to \eqref{eq:UpperBoundb1}-\eqref{eq:UpperboundEpsilon}, the following bound for regret that is satisfied with a probability of at least $1 - 9\delta$:
\begin{multline}\label{eq:regret_bound_action_based2}
    R_n^{(\mathbf{b})} \leq \\
    \sum_{a \in [k]}\sqrt{N_a}\overline{B}_{\theta}^\delta\left(n\right)\overline{B}_u^\delta\left(N_a\right)\left(\overline{B}_{(\mathbf{b})}^\delta\left(N_a\right) + \overline{\mathbf{B}}_{(\mathbf{b})}^\delta\left(N_a\right) \right) \\ + 2N_aB_{\beta}\left(s\right)\sqrt{g_{\Sigma_{z_t}}\left(\delta/N_a\right)}  + \sqrt{4N_a B_R^2 \log\left(1/\delta\right)}, 
\end{multline}

Setting $N_a \rightarrow n$, we arrive to the final bound for regret \eqref{eq:regret_bound_OFU} which is satisfied with a probability of at least $1 - 9\delta$. Based on \eqref{eq:regret_bound_OFU}, the upper bound for regret increases at the following rate \eqref{eq:regret_rate_increase} with a probability of at least $1 - 9\delta$. The term $\mathcal{O}\left(k\sqrt{(ms+1)n\log^3(n)}\right)$ is from model error $\hat{G}_a^t\left(s\right) - G_a\left(s\right)$ impacted by $\varepsilon_a^t$, $\mathcal{O}\left(kn\sqrt{(ms+1)\log^3(n)}\right)$ is from model error $\hat{G}_a^t\left(s\right) - G_a\left(s\right)$ impacted by $\beta_a^t$, $\mathcal{O}\left(k\sqrt{n}\right)$ is from the noise $\varepsilon_a^t$, and $\mathcal{O}\left(kn \sqrt{\log(n)}\right)$ are from the impacts of the bias $\beta_a^t$. 

\end{proof}

\begin{corollary}\label{corollary:regret_bound_2}
    Assume that $n > 0$ is finite. Let actions $a \in [k]$ be selected based on \eqref{eq:perturbation_method} where $u_a^{(\mathbf{e})}$ is defined as \eqref{eq:u_a_e} and $\mathbf{e}_a^t\left(s\right)$ is defined as \eqref{eq:e_def}. Regret $\sum_{t=1}^n X_t^* - X_t$ has the following  upper-bound with a probability of at least $1- 9\delta$: 
    \begin{multline}\label{eq:regret_bound_OFU_e}
        R_n^{(\mathbf{e})} \leq k\sqrt{n}\overline{B}_{\theta}^\delta\left(n\right)\overline{B}_u^\delta\left(n\right) \left(\overline{B}_{(\mathbf{b})}^\delta\left(n\right) + \overline{B}_{(\mathbf{e})}^\delta\left(n\right)\right)\\ + \sqrt{n\left(ms+1\right)N_a g_{\Sigma_{z_t}}\left(\delta/n\right)} B_s\left(n\right) \\ + 2 kn B_{\beta}\left(s\right)\sqrt{\log(n/\delta)} + k \sqrt{4 n B_R^2 \log(1/\delta)}, 
    \end{multline}
    where $\overline{B}_{(\mathbf{e})}^\delta\left(n\right)$ and $B_s\left(n\right)$ are defined as
    \begin{multline}\label{eq:large_e_bound_bold}
        \overline{B}_{(\mathbf{e})}^\delta\left(n\right)\triangleq \lambda B_G \sqrt{\frac{ms+1}{\lambda}}  \\ + \sqrt{2B_R^2\log\left(\frac{n}{\delta}\left(\frac{\lambda \left(ms+1\right) + n L_{\Theta}\left(\delta\right)}{\lambda \left(ms+1\right)}\right)^{\frac{ms+1}{2}}\right)},
    \end{multline}
    \begin{equation}\label{eq:s_bound_function}
        B_s\left(n\right) \triangleq n\sqrt{\frac{g_{\Sigma_{\Theta_t}}\left(\delta/n\right)}{\lambda}}. 
    \end{equation}

    Therefore, regret has an asymptotic rate of \eqref{eq:regret_rate_increase_no_bias}. Finally, if $s$ is set such that 
    \begin{equation}
        B_s\left(n\right)B_{\beta}\left(s\right) \leq \mathbf{B}_{\beta} , \label{eq:s_satisfactory}
    \end{equation}
    then regret has the following upper bound which is satisfied with a probability of at least $1-9\delta$
    \begin{multline}\label{eq:regret_bound_OFU_s}
        R_n^{(\mathbf{e})} \leq k\sqrt{n}\overline{B}_{\theta}^\delta\left(n\right)\overline{B}_u^\delta\left(n\right) \left(\overline{B}_{(\mathbf{b})}^\delta\left(n\right) + \overline{B}_{(\mathbf{e})}^\delta\left(n\right)\right)\\ + \sqrt{n\left(ms+1\right)g_{\Sigma_{z_t}}\left(\delta/N_{a^*}\right)}\mathbf{B}_{\beta} \\ + 2 kn B_{\beta}\left(s\right)\sqrt{g_{\Sigma_{z_t}}\left(\delta/N_{a^*}\right)} \\ + k\sqrt{4 n B_R^2 \log\left(1/\delta\right)}. 
    \end{multline}
\end{corollary}
\begin{proof}
    First, using Lemma \ref{lemma:optimal_value}, the following inequality is satisfied with a probability of $1-2\delta$:
    \begin{multline}
        G_{a^*}\left(s\right)^\top \Theta_t \overset{(a)}{\leq}  \hat{G}_{a^*}^t \left(s\right)^\top \Theta_t + u_{a^*}^{(\mathbf{e})}\\ + \sqrt{N_{a^*}\left(ms+1\right)g_{\Sigma_{z_t}}\left(\delta/N_{a^*}\right)}B_{\beta}\left(s\right)\sqrt{\Theta_t^\top V_{a^*}^t\left(s\right)^{-1}\Theta_t} \nonumber, 
    \end{multline}
    where in $(a)$ we used inequality \eqref{eq:UpperBoundb2} in Theorem \ref{theorem:regret_2}. Since action $a \in [k]$ is chosen instead of action $a^* \in [k]$, this implies that the following inequality is true
    \begin{equation}
        \hat{G}_{a^*}^t\left(s\right)^\top \Theta_t + u_{a^*}^{(\mathbf{e})}\leq  
        \hat{G}_{a}^t\left(s\right)^\top \Theta_t + u_{a}^{(\mathbf{e})}, \nonumber
    \end{equation}
    which implies that instantaneous regret has the following bound 
    \begin{multline}
        r_t^{(\mathbf{e})} \leq \hat{G}_{a}^t\left(s\right)^\top \Theta_t + u_{a}^{(\mathbf{e})} - G_a\left(s\right)^\top \Theta_t \\
        + \sqrt{N_{a^*}\left(ms+1\right)g_{\Sigma_{z_t}}\left(\delta/N_{a^*}\right)}B_{\beta}\left(s\right)\sqrt{\Theta_t^\top V_{a^*}^t\left(s\right)^{-1}\Theta_t} \\ + \beta_{a^*}^t - \beta_a^t + \varepsilon_{a^*}^t - \varepsilon_a^t. 
    \end{multline}
    
    For the terms $u_a^{(\mathbf{e})}$, $\hat{G}_{a}^t\left(s\right)^\top \Theta_t- G_a\left(s\right)^\top \Theta_t$, $\beta_{a^*}^t - \beta_a^t$, $\varepsilon_{a^*}^t - \varepsilon_a^t$, all their bounds are derived in Theorem \ref{theorem:regret_2}. However, for the term associated with $a^*$, we will use the worst case $V_a^t\left(s\right) \succeq \lambda I_{ms+1}$. Therefore, we have 
    \begin{align}
        & \sum_{a \in [k]}\sum_{t \in \mathscr{T}_a} \sqrt{\Theta_t^\top V_{a^*}^t\left(s\right)^{-1}\Theta_t} \nonumber \\
        & ~~~~~~~~~~~~~~~ \leq \sum_{a \in [k]}\sum_{t \in \mathscr{T}_a} \frac{\left\Vert \Theta_t\right\Vert_2}{\sqrt{\lambda}} \nonumber \\
        & ~~~~~~~~~~~~~~~ \overset{(b)}{\leq} \sum_{a \in [k]}\sum_{t \in \mathscr{T}_a} \sqrt{\frac{g_{\Sigma_{\Theta_t}}\left(\delta/n\right)}{\lambda}} \nonumber, 
    \end{align}
    \begin{equation}
        \Rightarrow \sum_{a \in [k]}\sum_{t \in \mathscr{T}_a} \sqrt{\Theta_t^\top V_{a^*}^t\left(s\right)^{-1}\Theta_t} \overset{(c)}{\leq} B_s\left(n\right) \nonumber, 
    \end{equation}
    where in $(b)$ we used Lemma \ref{lemma:subgaussian_half-normal} for $\left\Vert \Theta_t\right\Vert_2$ which is satisfied with a probability of at least $1-\delta$, in $(c)$ we used the fact that $\sum_{a \in [k]} \sum_{t \in \mathscr{T}_a} 1 = n$, and $B_s\left(n\right)$ is defined in \eqref{eq:s_bound_function}. Therefore, regret has the following upper bound that is satisfied with a probability of at least $1-9\delta$:
    \begin{multline}
        R_n^{(\mathbf{e})} \leq \sqrt{2N_{a^*}\left(ms+1\right)\log\left(N_{a^*}/\delta\right)}B_s\left(n\right)B_{\beta}\left(s\right) \\
        + \sum_{a \in [k]} \sqrt{N_a}\overline{B}_{\theta}^\delta\left(n\right)\overline{B}_u^\delta\left(N_a\right)\left(\overline{B}_{(\mathbf{b})}^\delta\left(n\right) + \overline{B}_{(\mathbf{e})}^\delta\left(n\right)\right)  \\ + 2 N_a B_{\beta}\left(s\right)\sqrt{g_{\Sigma_{z_t}}\left(\delta/N_a\right)} + \sqrt{4 N_a B_R^2 \log\left(1/\delta\right)}. 
    \end{multline}

    Setting $N_{a^*},N_a \rightarrow n$, we arrive at the regret bound \eqref{eq:regret_bound_OFU_e} which is satisfied with a probability of at least $1 -9\delta$. If we set $s$ such that inequality \eqref{eq:s_satisfactory} is satisfied, then regret has the upper bound \eqref{eq:regret_bound_OFU_s} which is satisfied with a probability of at least $1-9\delta$. Based on the set $s$, regret has an asymptotic rate of \eqref{eq:regret_rate_increase_no_bias}.     
\end{proof}

\section{Proof for Theorem \ref{theorem:s_cost}}\label{appendix:s_cost}

\begin{proof}
Using \eqref{eq:equality1}, we can express the term $G_a\left(s\right)^\top \Theta_t$ in \eqref{eq:Matrix_Form_2} as follows:
\begin{multline}\label{eq:reward_expression_cost_fun}
    X_t = \hat{G}_a^t \left(s\right)^\top \Theta_t + \left(\lambda G_a\left(s\right) - \mathbf{E}_{\mathscr{T}_a} \mathbf{O}_{\mathscr{T}_a}^\top\right) V_a^t\left(s\right)^{-1}\Theta_t \\ - \left(\mathbf{B}_{\mathscr{T}_a} \mathbf{O}_{\mathscr{T}_a}^\top\right) V_a^t\left(s\right)^{-1}\Theta_t + \beta_a^t + \varepsilon_a^t. 
\end{multline} 

Let there be the term $\nu \in [0,1]$. Using $\nu$, rearranging the terms in \eqref{eq:reward_expression_cost_fun}, and using the triangle inequality provides the following inequality:
\begin{multline} 
     X_t - \hat{G}_a^t  \left( s\right)^\top \Theta_t  -  \nu\left(\lambda G_a\left(s\right) - \mathbf{E}_{\mathscr{T}_a} \mathbf{O}_{\mathscr{T}_a}^\top \right) V_a^t\left(s\right)^{-1}\Theta_t  \\ - \varepsilon_a^t = \left( \left(1-\nu\right)\left(\lambda G_a\left(s\right)  - \mathbf{E}_{\mathscr{T}_a}\mathbf{O}_{\mathscr{T}_a} ^\top \right)\right) V_a^t\left(s\right)^{-1} \Theta_t\\  - \mathbf{B}_{\mathscr{T}_a} \mathbf{O}_{\mathscr{T}_a}^\top V_a^t\left(s\right)^{-1}\Theta_t + \beta_a^t, \nonumber
\end{multline}
\begin{multline}
     \Rightarrow \left\vert X_t - \hat{G}_a^t \left(s\right)^\top \Theta_t \right\vert \\
     + \nu\left\vert \left(\lambda G_a\left(s\right) - \mathbf{E}_{\mathscr{T}_a} \mathbf{O}_{\mathscr{T}_a}^\top\right) V_a^t\left(s\right)^{-1}\Theta_t \right\vert + \left\vert \varepsilon_a^t \right\vert \\ \geq \Bigg\vert \left(\left(1-\nu\right)\left(\lambda G_a\left(s\right) - \mathbf{E}_{\mathscr{T}_a}\mathbf{O}_{\mathscr{T}_a}^\top\right)\right) V_a^t\left(s\right)^{-1}\Theta_t\\ - \mathbf{B}_{\mathscr{T}_a} \mathbf{O}_{\mathscr{T}_a}^\top V_a^t\left(s\right)^{-1}\Theta_t + \beta_a^t \Bigg\vert, \nonumber
\end{multline} 

% {\color{red}Mehdi: The above implication is not correct. You cannot conclude equality; what you can say is that the  $\left\vert X_t - \hat{G}_a^t \left(s\right)^\top \Theta_t \right\vert+ \nu\left\vert \left(\lambda G_a\left(s\right) - \mathbf{E}_{\mathscr{T}_a} \mathbf{O}_{\mathscr{T}_a}^\top\right) V_a^t\left(s\right)^{-1}\Theta_t \right\vert + \left\vert \varepsilon_a^t \right\vert\geq \Bigg\vert \left(\left(1-\nu\right)\left(-\lambda G_a\left(s\right) + \mathbf{E}_{\mathscr{T}_a}\mathbf{O}_{\mathscr{T}_a}^\top\right)\right) V_a^t\left(s\right)^{-1}\Theta_t+ \mathbf{B}_{\mathscr{T}_a} \mathbf{O}_{\mathscr{T}_a}^\top V_a^t\left(s\right)^{-1}\Theta_t + \beta_a^t \Bigg\vert$. Unless, I am missing something.}

The upper bound of the term
\begin{equation*}
    \left\vert \left(\lambda G_a\left(s\right) - \mathbf{E}_{\mathscr{T}_a} \mathbf{O}_{\mathscr{T}_a}^\top\right) V_a^t\left(s\right)^{-1}\Theta_t \right\vert,
\end{equation*}
is the solution of the following optimization problem that is satisfied with a probability of at least $1-\delta$
\begin{equation}\label{eq:minmax_opt}
    \begin{array}{cc}
        \underset{\tilde{G} \in \mathbb{R}^{(ms+1)\times 1}}{\max} & \tilde{G}^\top \Theta_t - \hat{G}_a^t\left(s\right)^\top \Theta_t \nonumber \\
        \mbox{ s.t. } & \tilde{G} \in \mathscr{E}_{a,s}
    \end{array}, 
\end{equation} 
where $\mathscr{E}_{a,s}$ is defined to be 
\begin{align}
    \mathscr{E}_{a,s}& \triangleq \left\{\left\Vert \hat{G}_a^t\left(s\right) - \Tilde{G}\right\Vert_{V_a^t\left(s\right)} \leq \mathbf{e}_a^t\left(\delta,s\right)\right\}. 
\end{align}
 
Using Lemma \ref{lemma:optimal_value}, the solution of \eqref{eq:minmax_opt} is 
\begin{equation}
    \mathbf{e}_a^t\left(\delta,s\right)\sqrt{\Theta_t^\top V_a^t\left(s\right)^{-1}\Theta_t} \nonumber. 
\end{equation}

Therefore, no further effort is need to prove that the inequality \eqref{eq:upper_bound_2} is satisfied with the probability of at least $1-\delta$; this completes the proof. 
\end{proof}

\section{Proof for Theorem \ref{theorem:lower_bound_discrete}}\label{appendix:lower_bound}

\begin{proof}
    Let there be the expected regret $R_n$ found in \eqref{eq:regret}. 
    \begin{align}
        R_n & = \sum_{t=1}^n \mathbb{E}\left[X_t^* - X_t\right] \nonumber \\
        & = \sum_{t=1}^n \mathbb{E}\left[\left\langle c_{a_t^*} - c_{a_t} , z_t \right\rangle \mid \left\langle c_a - c_{a'}, z_t \right\rangle \geq 0, a_t^* = a\right] \nonumber, 
    \end{align}
    where in $(a)$ we used the Markov inequality which is satisfied with a probability of at most $\delta$ where $\delta \in (0,1)$. 
    \begin{multline}\label{eq:regret_breakdown}
        \Rightarrow R_n \overset{(a)}{=} \\ \sum_{t=1}^n \sum_{c_a,c_{a'} \in \mathcal{A}} \mathbb{E}_{z_t}\left[\left\langle c_{a_t^*} - c_{a_t} , z_t \right\rangle \mid a_t^*, a_t\right] \\ \cdot P\left(a_t = a'\mid a_t^* = a\right), 
    \end{multline}
    where in $(a)$ we used the Law of Total of Expectation. Let us assume at round $t$ that the action selected by the \textit{Oracle} \eqref{eq:Kalman_Oracle} is $a \in \mathcal{A}$. We want to find the probability that the \textit{Oracle} \eqref{eq:Kalman_Oracle} chooses an action $c_{a'} \in \mathcal{A}$ such that $a' \neq a$. The event of this occurring is based on the following sets
    \begin{align}
        \mathcal{E}_t^a & \triangleq \cap_{c_{a'} \in \mathcal{A}} \left\{\left\langle c_a - c_{a'}, z_t\right\rangle > 0\right\} \nonumber \\
        \Tilde{\mathcal{E}}_t^a & \triangleq \cap_{c_{a'} \in \mathcal{A}} \left\{\left\langle c_a - c_{a'}, \Tilde{z}_t\right\rangle > 0\right\} \nonumber . 
    \end{align}

    Next, we want to find the distribution of 
    \begin{equation}\label{eq:dist_opt}
        \left(\left\langle c_a - c_{a'}, \Tilde{z}_t\right\rangle,\left\langle c_a - c_{a'}, z_t\right\rangle\right). 
    \end{equation}
    $\left(\left\langle c_a - c_{a'}, \Tilde{z}_t\right\rangle,\left\langle c_a - c_{a'}, z_t\right\rangle\right)$. Recall that in the Kalman filter the state $z_t$ can be expressed as $z_t = \Tilde{z}_t+\Tilde{e}_t$. Therefore the joint distribution of \eqref{eq:dist_opt} is based on the following bounds of $P\left(\mathcal{E}_t^{a_j}\mid \Tilde{\mathcal{E}}_t^{a_i} \right)$:
    \begin{align}
        & P\left(\mathcal{E}_t^{a_j}\mid \Tilde{\mathcal{E}}_t^{a_i} \right) \nonumber \\
        & ~~ = \int_{\mathbb{R}_{+}^{k-1}}\int_{\mathbb{R}_{+}^{k-1}} P\left(A_i \Tilde{z}_t + \Pi_{i|j} \zeta = \Tilde{\zeta} \mid A_j z_t = \zeta\right)  d\zeta d\Tilde{\zeta} \nonumber \\
        & ~~ = \int_{\mathbb{R}_{+}^{k-1}}\int_{\mathbb{R}_{+}^{k-1}} \frac{\exp\left(-\frac{\left(\Tilde{\zeta} - \Pi_{i|j} \zeta\right)^\top \Tilde{\Sigma}_{i|j}^{-1}\left(\Tilde{\zeta} - \Pi_{i|j} \zeta \right)  }{2}\right)}{\sqrt{\left(2\pi\right)^{2k-2} \left\vert \Tilde{\Sigma}_{i|j}\right\vert }} d\zeta d\Tilde{\zeta} \nonumber \\
        & ~~ \overset{(b)}{=} \int_{\mathbb{R}_{+}^{2k-2}}\frac{\exp\left(-\frac{1}{2}\Vec{\zeta}^\top \Psi_{i|j}   \Vec{\zeta}\right)}{\sqrt{\left(2\pi\right)^{2k-2} \left\vert \Tilde{\Sigma}_{i|j}\right\vert }} d\Vec{\zeta} \nonumber \\
        & ~~ = \int_{\mathbb{R}_{+}^{2k-2}}\frac{\exp\left(-\frac{1}{2}\mbox{tr}\left(\Vec{\zeta}\Vec{\zeta}^\top \Psi_{i|j}   \right)\right)}{\sqrt{\left(2\pi\right)^{2k-2} \left\vert \Tilde{\Sigma}_{i|j}\right\vert }} d\Vec{\zeta} \nonumber \\
        & ~~ \geq \int_{\mathbb{R}_{+}^{2k-2}}\frac{\exp\left(-\frac{1}{2}\mbox{tr}\left(\Vec{\zeta}\Vec{\zeta}^\top \right) \mbox{tr}\left(\Psi_{i|j}\right)   \right)}{\sqrt{\left(2\pi\right)^{2k-2} \left\vert \Tilde{\Sigma}_{i|j}\right\vert }} d\Vec{\zeta} \nonumber \\
        & ~~ = \frac{\prod_{s=1}^{2k-2}\int_0^\infty \exp\left(-\frac{\Vec{\zeta}[s]^2}{2\mbox{tr}\left(\Psi_{i|j}\right)^{-1}} \right) d\Vec{\zeta}[s]}{\sqrt{\left(2\pi\right)^{2k-2} \left\vert \Tilde{\Sigma}_{i|j}\right\vert }}  \nonumber \\
        & ~~ = \frac{\prod_{s=1}^{2k-2} \sqrt{\frac{2\pi}{\mbox{tr}\left(\Psi_{i|j}\right)}}}{\sqrt{\left(2\pi\right)^{2k-2} \left\vert \Tilde{\Sigma}_{i|j}\right\vert }}  \nonumber, 
    \end{align}
    \begin{equation}
        \Rightarrow P\left(\mathcal{E}_t^{a_j}\mid \Tilde{\mathcal{E}}_t^{a_i} \right) \geq \frac{1}{\sqrt{\mbox{tr}\left(\Psi_{i|j}\right)^{2k-2} \left\vert \Tilde{\Sigma}_{i|j}\right\vert }} \nonumber ,
    \end{equation}
    where in $(b)$ we replaced $\tilde{\Sigma}_{i|j}$ with $\Psi_{i|j}$. Finally, we need the expectation $\mathbb{E}_{z_t}\left[\left\langle c_{a_t^*} - c_{a_t} , z_t \right\rangle \mid a_t^*, a_t\right]$. We know that based on the definition of $a_t^*$, $\left\langle c_{a_t^*} - c_{a_t} , z_t \right\rangle > 0$. We also know that $z_t$ is a normally distributed random variable $z_t \sim \mathcal{N}\left(\mathbf{0},Z\right)$ where $Z = \Gamma Z \Gamma^\top + Q$. Therefore, the conditional expectation is 
    \begin{multline}
        \mathbb{E}_{z_t}\left[\left\langle c_{a_t^*} - c_{a_t} , z_t \right\rangle \mid a_t^*, a_t\right] \\ = \sqrt{\frac{2\left(c_{a_t^*} - c_{a_t} \right)^\top Z \left(c_{a_t^*} - c_{a_t} \right)}{\pi}}\nonumber. 
    \end{multline}

    Therefore, regret for the \textit{Kalman Oracle} is \eqref{eq:regret_lower_bound_discrete}. 

\end{proof}

\bibliographystyle{IEEEtran}
\bibliography{IEEEabrv,autosam}{}

\end{document}